%% file: minkowski.tex
\RequirePackage{fix-cm}
\documentclass[smallcondensed]{svjour3}      
\smartqed  
\usepackage{graphicx}
\usepackage{mathptmx}      
\usepackage[cal=cm]{mathalfa}
\usepackage{hyperref} 
\usepackage{amssymb, amsmath}
\usepackage{float}
\usepackage{fancyvrb}
\usepackage[fancyvrb]{listings}
\usepackage{xcolor}
\usepackage{etoolbox}
\usepackage[disable,textsize=footnotesize,textwidth=5cm]{todonotes}
\usepackage[normalem]{ulem}
\usepackage{color,soul}
\usepackage{epigraph}
\usepackage{caption}
\usepackage{subcaption}
\usepackage{standalone}
\usepackage{tikz}
\usetikzlibrary{arrows,calc}
\definecolor{tikz-red}{rgb}{0.5019607843137255,0,0}
\definecolor{tikz-darkblue}{rgb}{0,0.2,0.6}
\definecolor{tikz-blue}{rgb}{0.49019607843137253,0.49019607843137253,1}
\definecolor{tikz-gray}{rgb}{0.5,0.5,0.5}
\newcommand{\E}{\mathcal{E}}
\renewcommand{\P}{\mathcal{P}}
\newcommand{\spray}{\ensuremath{\textrm{SPRAY}}}
\newcommand{\ord}[3]{\ensuremath{[#1\;#2\;#3]}}

\newcommand{\chain}[2][\;\dots\;]{%
 \ensuremath{\left[%
  \def\nextitem{\def\nextitem{#1}}
  \renewcommand*{\do}[1]{\nextitem##1}
  \docsvlist{#2}
 \right]}%
}
\makeatletter
\newcommand{\customlabel}[2]{%
   \protected@write \@auxout {}{\string \newlabel {#1}{{#2}{\thepage}{#2}{#1}{}} }%
   \hypertarget{#1}{#2}
}
\makeatother

\spnewtheorem*{axiom}{Axiom}{\bf}{\rm}
\spnewtheorem*{theorem*}{Theorem}{\bf}{\it}  

\DeclareTextSymbol{\textbackslash}{T1}{92}

\newcommand{\setN}{{\mathord{\mathbb N}}}
\newcommand{\setQ}{{\mathord{\mathbb Q}}}
\newcommand{\setR}{{\mathord{\mathbb R}}}

\newcommand{\lsemantics}{\mathopen{\lbrack\mkern-3mu\lbrack}}
\newcommand{\rsemantics}{\mathclose{\rbrack\mkern-3mu\rbrack}}

\newcommand{\jtodo}[2][]{\todo[color=red!70,#1]{#2}}
\newcommand{\jeptodo}[2][]{\todo[color=green!70,#1]{#2}}
\input{in/lisa}
\hypersetup{
    colorlinks,
    linkcolor={red!50!black},
    citecolor={blue!50!black},
    urlcolor={blue!80!black}
}
\journalname{Journal of Automated Reasoning}
\begin{document}

\title{Towards Formalising Schutz' Axioms for Minkowski Spacetime in Isabelle/HOL
}
\subtitle{}


\author{Richard Schmoetten \and Jake E. Palmer \and Jacques D. Fleuriot}


\institute{
	R. Schmoetten \at
		Artificial Intelligence and its Applications Institute,
		The University of Edinburgh\\
		\email{s1311325@sms.ed.ac.uk}
\and
	J. E. Palmer \at
		Artificial Intelligence and its Applications Institute,
		The University of Edinburgh\\
		\email{jake.palmer@ed.ac.uk}
\and 
	J. D. Fleuriot \at
		Artificial Intelligence and its Applications Institute,
		The University of Edinburgh\\
		\email{jdf@ed.ac.uk}
}


\maketitle

\begin{abstract}
Special Relativity is a cornerstone of modern physical theory. While a standard coordinate model is well-known and widely taught today, several alternative systems of axioms exist. This paper reports on the formalisation of one such system which is closer in spirit to Hilbert's axiomatic approach to Euclidean geometry than to the vector space approach employed by Minkowski.
We present a mechanisation in Isabelle/HOL of the system of axioms as well as theorems relating to temporal order. Proofs and excerpts of Isabelle/Isar scripts are discussed, particularly where the formal work required additional steps, alternative approaches, or corrections to Schutz' prose.
\end{abstract}

\section{Introduction}
\label{sec:intro}
\input{in/intro}
\todo{formatting at the end: make sure proofs are not split across pages; make sure lstinline doesn't lead to overfull hboxes (seen in warnings, words extending into right margin); noindent after listings?; see if more places need linking of text/Isa/Schutz (e.g. by filling in variable names); mathpmx fonts; indent for proof env?;}

\section{Background}
\label{sec:background}
    \subsection{Formalisation in Special Relativity}
        \label{sec:bgr:physics}
        \input{in/bgr-physics}
    \subsection{Axiomatic Geometries}
        \label{sec:bgr:geometry}
        \input{in/bgr-geometry}
    \subsection{Isabelle/HOL}
        \label{sec:bgr:isa}
        \input{in/bgr-isa}

\section{Axioms}
\label{sec:axioms}
\input{in/axioms}

\section{Formalisation: Temporal Order on a Path}\label{sec:ch3}
We have formalised all of Schutz' results from Chapter 3 (\emph{Temporal Order on a Path}) of his monograph, except for Theorem~12 (\emph{Continuity}; see Section~\ref{sec:conclusion} for a short discussion). In many cases, his statements had to be extended or amended to pass Isabelle's unforgiving scrutiny. In what follows, rather than giving formal proofs for all of these results, we sketch the proofs given by Schutz and highlight interesting features of their formalisation. We refer to the Isabelle proof document%
\footnote{To be accessed at \href{https://github.com/rhjs94/schutz-minkowski-space.git}{https://github.com/rhjs94/schutz-minkowski-space}.}
for the complete proof script, and the original monograph \cite{schutz1997} for sometimes more extensive prose, when we do not reproduce it.

We endeavour to present proof procedures at a comfortable level of detail. Fairly often, extra steps required in Isabelle are obvious to the inspecting reader; usually their omission does not obscure the flow of the overall argument. We therefore employ ``snipping'' rather freely. We denote by \lstinline|<proof>| a proof that was cut, but exists in the associated proof script. The notation \lstinline|(*...*)| is used for cutting away multiple not necessarily related lines, or even just a part of a line. This relaxation is possible because we trust the Isabelle verification of our proof: if one wanted to verify all the statements in this paper, one could simply make sure they exist in the Isabelle theory, identify the introduced axioms, and let Isabelle check the entire file. Regardless of snipping, all results presented are completed and accepted by Isabelle.

The following section is ordered as in Schutz' monograph, and this structure is reflected in the formal proof document as well.

    \subsection{Order on a finite chain}
        \label{sec:order-fin-chain}
        \input{in/thms-3.1}
    \subsection{First collinearity theorem}
        \label{sec:collinearity1}
        \input{in/thms-3.2}
    \subsection{Boundedness of the unreachable set}
        \label{sec:unreach-bounded}
        \input{in/thms-3.3}
    \subsection{Prolongation}
        \label{sec:prolongation}
        \input{in/thms-3.4}
    \subsection{Second collinearity theorem}
        \label{sec:collinearity2}
        \input{in/thms-3.5}
    \subsection{Order on a path}
        \label{sec:order-path}
        \input{in/thms-3.6a}  
        \input{in/thms-3.6b}  
    \subsection{Continuity and the monotonic sequence property}
        \label{sec:continuity}
        \input{in/thms-3.7}
        \setcounter{theorem}{12}
    \subsection{Connectedness of the unreachable set}
        \label{sec:unreach-connected}
        \input{in/thms-3.8}
        
\section{Conclusion and Future Work}\label{sec:conclusion}
Our formalisation of temporal order on paths in Schutz' axiomatic Minkowski space is over nine thousand lines long. Schutz' admirably detailed account (for prose) covers 22 pages. Estimating thirty lines on each page, this leaves us with a de Bruijn factor \cite{wiedijk2000,debruijn1994a} of roughly 14. This is not exceptional: while many formalisations only report de Bruijn factors as low as 3 to 6, values above 20 can be found \cite{dzamonja2020}. We also note that a recent batch of simplifications and rewritten proofs has cut our formalisation by about eight hundred lines, so this estimated factor may be further reducible.

One should note that the axiomatisation by itself would have a factor of only around 4. The thirteen formalised theorems and their proofs, together with most added intermediate lemmas, have de Bruijn factor of roughly 23. This, we estimate, is largely due to the later proofs of the chapter relying more strongly on Schutz' geometric intuition, the validation of which in the context of his axioms is the main goal of Chapter~3. Thus our formal constructions had to become more and more elaborate (the prime example is our collection of WLOG lemmas), and supported by lengthy existence proofs omitted in the original prose.

Several required lemmas were not stated in the original text, most notably, in the proof of Theorem~\ref{thm:10}. Theorem~\ref{thm:11} saw a minor correction to the statement, while Theorem~\ref{thm:14} required WLOG-style lemmas to avoid a large number of case splits. Refinement of these WLOG-lemmas is one avenue to pursue in future work, as it could prove useful in a large number of formalisations beyond ours, and captures a kind of mathematical reasoning device employed frequently, and to great effect, in prose. A similar investigation could try to link results of symmetry, such as a chain being reversible (\lstinline|chain_sym|), and sufficient subgoals for avoiding case splits. Ideally, such subgoals might be generated automatically based on the symmetry considered.

Our formalisation covers the third chapter of Schutz' monograph, with the exception of Theorem~12, the Continuity Theorem. This is the only result of the chapter that intensely relies on working with infinite chains, and thus falls outside the scope of this paper. 
Avoiding continuity in a first effort to formalise a geometry has precedence, for example in the work of Meikle and Fleuriot, which largely focuses on the first three groups of axioms of Hilbert's \emph{Grundlagen} (continuity appears in the fifth), or the investigation of the first four groups in Coq by Braun and Narboux \cite{meikle2003,braun12}. We do note here that we have made some progress towards mechanising Theorem~12 though. In particular, we have formalised its first part, which partitions any path into two rays. The second part attempts to show continuity formulated in a manner analogous to the construction of $\setR$ as Dedekind cuts of $\setQ$. This is not yet formalised. We plan to continue this work, establishing paths as continuous. This may require a stronger axiom of continuity, if Schutz' proof is shown to be insufficient.



While Schutz insists upon the independence of his axiomatic system, even mentioning it in the title of his monograph, future work on this formalisation may emphasise this property less, in favour of easier, more modular organisation, as well as more succinct definitions and axioms. If the great reward of the quest for independence was, as Schutz claims, a set of intuitive and clear axioms, then it seems justified to step away from strict independence towards a clearer organisation in Isabelle's locales. Another potential aspect of further work lies in trying to apply automation tools from similar formalisations, such as automated tactics to translate from problems of ordering on events to natural numbers and proof discovery tools \cite{scott2015,scott2011}.

\subsection{Final Remarks}
The programme of axiomatisation of foundational physics goes back at least to Hilbert's sixth problem in 1900. Once a candidate system of axioms is constructed, its formalisation in a proof assistant such as Isabelle is a natural continuation, both for the certainty of correctness it offers, and for the analysis (and maybe even automation) of the employed reasoning. A geometrically-inspired system such as the one of Schutz can be a valuable link between geometric intuition and physical theorems. In our case, several axioms similar to those of Hilbert's \emph{Grundlagen def Geometrie} meet an order-theoretic approach that may be compared to modern ideas for the foundations of physics \cite{knuth2014,knuth2017,goyal2010}.

Thus this formalisation contributes not only a study of the foundations of Special Relativity, but may provide a link between approaches from synthetic geometry and foundational physics. We hope that future work will not only extend our mechanisation to include and clarify more of Schutz' results, but will also investigate more general aspects of automated reasoning in axiomatic physics.

\bibliographystyle{spmpsci}      
\bibliography{schutz,manual}   


\end{document}

%% file: in/lisa.tex
\definecolor{isa-rangecomment}{RGB}{139,0,0}
\definecolor{bgr-lightgrey}{gray}{0.98}
\definecolor{isa-kwcolor-2}{RGB}{27,158,119}
\definecolor{isa-kwcolor-3}{RGB}{217,95,2}
\definecolor{isa-kwcolor-4}{RGB}{117,112,179}
\colorlet{isa-proofsnip}{isa-kwcolor-3}
\lstdefinelanguage{Isar}%
{
  morekeywords=[2]{
    _first,
    theory,begin,end,types,datatype,consts,defs,primrec,
    syntax,translations,apply,
    lemma,theorem,corollary,definition,inductive,where,locale,context,abbreviation,
    done,next,proof,and,qed,
    assumes,fixes,shows,case,goal,obtains,defines,OF,in,
    fixed variables,
    _last
  },
  keywordstyle={[2]\bfseries\color{isa-kwcolor-2}},
  morekeywords=[3]{
  	sorry,oops
  },
  keywordstyle={[3]\bfseries\color{isa-kwcolor-3}},
  morekeywords=[4]{
    then,this,thus,hence,have,assume,obtain,moreover,fix,show,by,using,unfolding,
    ultimately,let,if,for,else,consider
  },
  keywordstyle={[4]\bfseries\color{isa-kwcolor-4}},
  sensitive=true,
  morecomment=[l]{-- },
  morecomment=[s]{(*}{*)},
  literate=
  	{(*...*)}{{\textcolor{isa-proofsnip}{\textbf{...}}}}3
  	{<proof>}{{\textcolor{isa-proofsnip}{\textbf{<proof>}}}}7
    {\\<not>}{{$\neg$}}1
    {\\<times>}{{$×$}}1                 
    {\\<Rightarrow>}{{$\Rightarrow$}}2%
    {\\<equiv>}{{$\equiv$}}1
    {~=}{{$\not=$}}1                         
    {\\<rightleftharpoons>}{{$\rightleftharpoons$}}2
    {\\<exists>}{{$\exists$}}1
    {\\<forall>}{{$\forall$}}1
    {\\<in>}{{$\in$}}1
    {\\<notin>}{{$\notin$}}1
    {\\<union>}{{$\cup$}}1
    {\\<Union>}{{$\bigcup$}}2
    {\\<inter>}{{$\cap$}}1
    {\\<emptyset>}{{$\emptyset$}}1
    {\\<theta>}{{$\theta$}}1
    {\\<Theta>}{{$\Theta$}}1
    {\\<lambda>}{{$\lambda$}}1
    {\\<E>}{{$\E$}}1
    {\\<subseteq>}{{$\subseteq$}}1
    {\\<ge>}{{$\geq$}}1
    {\\<le>}{{$\leq$}}1
    {\\<noteq>}{{$\neq$}}1
    {\\<P>}{{$\P$}}1
    {\\<parallel>}{{$\parallel$}}1
    {\\<Longrightarrow>}{{$\Longrightarrow$}}3
    {\\<longrightarrow>}{{$\longrightarrow$}}3
    {\\<leftrightarrow>}{{$\leftrightarrow$}}2
    {\\<longleftrightarrow>}{{$\longleftrightarrow$}}3
    {\\<lbrakk>}{{$\lsemantics$}}1
    {\\<rbrakk>}{{$\rsemantics$}}1
    {\\<and>}{{$\land$}}1
    {\\<or>}{{$\lor$}}1
    {\\<bigwedge>}{{$\bigwedge$}}1
    {\\<And>}{{$\bigwedge$}}1
    {\\<triangle>}{{$\triangle$}}1
    {\\<open>}{{$\langle$}}1
    {\\<close>}{{$\rangle$}}1
    {\\<\^cancel>}{{\%}}1
    {\\<\^sub>}{{\_}}1
    {ü}{"u}1 {ä}{"a}1 {ö}{"o}1
    {Ü}{"U}1 {Ä}{"A}1 {Ö}{"O}1 {ß}{{ss}}1
}[keywords,comments,strings]%

\renewcommand{\lstinline}{\oldlstinline[basicstyle=\ttfamily\normalsize]}

%% file: in/intro.tex
Formal foundations are a recently re-emerging trend in modern physics. While philosophical, mathematical, and empirical studies were inseparably entwined in antiquity, formal mathematics and physical science drifted apart in the eighteenth and nineteenth centuries \cite{suppes1968}.

The mathematical deduction employed for example in Ptolemy's \textit{Harmonics} is taken to be almost divine. Thus he considers ``arithmetic and geometry, as instruments of indisputable authority'' \cite[p.~507]{bernard2010}.
In contrast, the main physical theories of the twentieth century were developed as physics first, and retro-fitted with rigorous mathematical foundations later.
An example particularly relevant to this work is that of Special Relativity (SR) \cite{gourgoulhon2013g}. 
The comprehensive mathematical treatment given by Minkowski \cite{minkowski1908} was at first dismissed as unnecessarily complicated
\cite{einstein1908}.
Early work on axiomatising SR (e.g. by Robb \cite{robb1936}) went largely unnoticed by the physical research community.

But the search for a formal foundation to modern physics gained support in the second half of the twentieth century. Philosophical essays \cite{suppes1968}, the successes of the new mathematical quantum and relativity theories \cite{schrodinger1926,born1926}, and increasing interest by the mathematical community, all contributed to works ranging from differential geometry and General Relativity (GR) to the Wightman axioms in particle physics \cite{streater2000}.

We will present here a mechanisation of an axiom system for Minkowski spacetime, the main ingredient of the theory of SR, given by Schutz in 1997 \cite{schutz1997}. To this end, we use the proof assistant Isabelle/HOL, briefly introduced in Section~\ref{sec:background}. We then proceed to an exhibition of the axioms in Section~\ref{sec:axioms}, and describe some of our mechanised lemmas and theorems in Section~\ref{sec:ch3}.%
\footnote{The formalisation can be accessed at \href{https://github.com/rhjs94/schutz-minkowski-space}{https://github.com/rhjs94/schutz-minkowski-space}.}

%% file: in/bgr-physics.tex

Several axiom systems have been proposed for Minkowski spacetime.
Schutz himself proposes several iterations, starting with a formulation based on primitive particles and the binary \textit{signal} relation in 1973 \cite{schutz1973}. The next iteration in 1981 replaces \textit{signals} with a binary \textit{temporal order} relation, and light signals become an entirely derived notion, whose existence is proven, not assumed \cite{schutz1981}. It is the final axiom system, published in a monograph in 1997, that is of primary interest to us: it contains many of the axioms of earlier systems as theorems, while also boasting the property of independence (see Sec.~\ref{sec:axioms} for details).
Systems formulated by Szekeres \cite{szekeres1968} and Walker \cite{walker1959} also rely on undefined bases and axioms inspired by physical intuition, and Schutz cites them as direct predecessors to his work.
Another early approach is that of Robb \cite{robb1936}, based on events and an ordering relation, and continued by Mundy \cite{mundy1986,mundy1986a}.
A first-order alternative to Schutz is given by Goldblatt \cite{goldblatt1989,goldblatt2012}, who relies on a relation of orthogonality in addition to the betweenness Schutz employs in his system of 1997.

More recently, an extension of Tarski's Euclidean ideas using Goldblatt's approach to Minkowski spacetime was given by Cocco and Babic \cite{cocco2021}. Their system is mostly formulated in first-order logic, but with a second-order continuity axiom in order to show the usual four-dimensional Minkowski spacetime is a model.
A flexible first-order system of axioms describing several different theories of relativity was given by Andr\'eka et al. \cite{andreka2011,andreka2013}. Notably, there exists a mechanisation of this approach in Isabelle/HOL by Stannett and N\'emeti \cite{stannett2014}. In contrast to what we propose here, Stannett and N\'emeti assume an underlying coordinate formulation and use first-order axioms, while Schutz' system is second-order, and his Isomorphism Theorem linking it to the usual coordinate model is one of his final results.

%% file: in/bgr-geometry.tex

Geometry is arguably the oldest discipline to have seen successful axiomatisation in the form of Euclid's \textit{Elements} \cite{heath1956}. Over two millennia later, Hilbert's \textit{Grundlagen der Geometrie}
\cite{hilbert1950} built on Euclid to propose a new, self-contained system of axioms using modern logical concepts such as undefined notions (in contrast to Euclid's primitive definitions). Many alternative Euclidean systems have been postulated and examined since then.
Schutz acknowledges clear parallels between several of his theorems and those of Veblen \cite{veblen1904}, whose axioms for Euclidean geometry replace Hilbert's primitives (points, lines, planes, and several relations between them) to use only points and a single relation.
Tarski's system of elementary Euclidean geometry \cite{tarski1959} is influential too: points as well as two undefined relations are his only primitive notions. His axioms can be formulated in primitive notions only, using first-order logic (with identity and using an axiom schema).
Schutz \cite{schutz1997} similarly strives for simplicity, though his continuity axiom is second-order, and while a line-like primitive exists, only a single undefined relation is required.

\subsubsection{Mechanisation in Geometry}\label{sec:mech-geom}
Several axiomatic approaches to geometry have been (at least partially) formalised in Isabelle/HOL. Hilbert's \textit{Grundlagen} has seen work in Isabelle by Meikle, Scott and Fleuriot \cite{meikle2003,scott2008}, and further investigation of both the axioms and tools for their study in HOL Light \cite{scott2011,scott2015}. Tarski's axiom system was investigated by Narboux in Coq \cite{narboux2007}, and its independence verified in Isabelle by Makarios \cite{makarios2012}. Geometric formalisations also exist e.g. for projective geometry in Coq \cite{magaud2011} and again for Tarski's geometry in Mizar \cite{grabowski2016}. We refer to a recent review for a more comprehensive picture \cite{narboux2018}.

Our formalisation bears some similitude to the above work on Hilbert's \emph{Grundlagen} in a number of respects since several of Schutz' axioms originate in the \emph{Grundlagen} (see Section~\ref{sec:axioms}). For example, our definition of chains (Section~\ref{sec:axioms:chains}), one of the most fundamental constructs in this paper, relies on an adapted definition due to Scott's work on the \emph{Grundlagen} in HOL Light \cite{scott2015}.
As another example, we employ the same weakening of Schutz' Axiom O3 that can be found in Scott's formalisation of Hilbert's Axiom II.1. Scott also finds a result very similar to our \lstinline|chain_unique_upto_rev| (see Section~\ref{sec:unreach-connected}): while he obtains it from a remark of Hilbert's \cite[Section~6.7.2]{scott2015}, we derived it by necessity in an early version of our proof of Theorem~\ref{thm:13}, and found the correspondence only later.
Notice the formalisations of Hilbert's \emph{Grundlagen} cited here focus on the first three groups of axioms, which exclude the parallel and continuity axioms.


%% file: in/bgr-isa.tex
Computer-based theorem proving, verification, and proof exploration is the dominant area of automated reasoning today. A breakthrough development for the field was Scott's work on LCF \cite{scott1993a}, a typed version of the $\lambda$-calculus, and the subsequent construction of an interactive theorem prover of the same acronym by Gordon, Milner and Wadsworth \cite{gordon1979}.
Isabelle is a generic proof assistant which continues the LCF-style of automated reasoning \cite{wenzel2008,paulson2019}. Its generic meta-logic (the simple type system responsible for validity checking) supports multiple instances of object logic: we will be using higher order logic (HOL), but instances for e.g. first-order logic (FOL) and ZFC set theory exist.

We review several salient aspects of Isabelle below, and give a brief introduction to proof reading and writing.



\subsubsection{Automation and Readability}

\epigraph{A proof is a repeatable experiment in persuasion.}{Jim Horning}

Considering the above quote, the advantage of computer assistance in logical and mathematical proof is clear. Using Isabelle (for example), we can write a proof of any (provable) theorem, and provided our readers are convinced of the soundness%
\footnote{The consistency of theorem provers is its own research field \cite{kuncar2017}.}
of Isabelle's trusted kernel, they can take the theorem as fact without manually verifying the proof.
A famous and well-popularised success of computer-verified mathematics is the \textit{Flyspeck} project \cite{hales2015}.
A computer-assisted proof of the Kepler conjecture was submitted for review in 1998, but only published (without the reviewers' complete certification) in 2006 \cite{hales2006,lagarias2011a}.
The \textit{Flyspeck} project is a twelve-year effort to formalise this proof,
accepted to a mathematical journal in 2017. 

Even if a proof is certified and trusted, it is often still instructive to read through it. One may identify methods to be used in similar problems, or generalised to unrelated areas of inquiry; intuition is built for the behaviour of the mathematical entities manipulated throughout the proof.
Readability is therefore important, particularly for proofs as verbose as those often found in mechanisations.
Isabelle provides us with the language Isar (\textbf{I}ntelligible \textbf{s}emi-\textbf{a}utomated \textbf{r}easoning) \cite{wenzel1999} that can be used for proofs that are both human-readable and supported by automatic solvers. Isar proofs merge the forward reasoning common in mathematical texts and natural for human readers to follow, and the backwards reasoning often useful in exploring possible avenues for a proof to be completed (see the next section for a glimpse of Isar).

Several tools for proof discovery come with the Isabelle distribution. In particular, the umbrella tool \textit{sledgehammer} \cite{paulson2010} automatically chooses a range of (several hundred) facts to pass to different first-order solvers (both resolution and SMT provers), and, if successful, provides a reconstruction of the automatic proof in Isabelle/HOL. In practice, automatic proof discovery is useful, but sometimes struggles to justify steps that seem obvious to the reader, or returns proofs relying on highly unexpected facts. This may be due to the complexity of some of our definitions, or difficulty in reductions to first-order logic. It has also sometimes led to corrections to axioms.

\subsubsection{Proofs and Isar}\label{sec:bgr:isaind}


Working in Isabelle/HOL (and Isar) is a mix of meta- and object-level reasoning.
This is best looked at through an example: we use a lemma named \hbox{\lstinline|no_empty_paths|} from our current formalisation. We are only interested in the formalism and method for now. Sec.~\ref{sec:axioms:unreach} will provide some context.

Meta-logic in Isabelle can be part of the inner syntax (between double quotes, e.g. \lstinline|\<lbrakk>...\<rbrakk>| for assumptions and \lstinline|\<Longrightarrow>| for meta-implication) or outer syntax (e.g. \lstinline|assumes|, \lstinline|shows|).
We announce the statement of a top-level fact requiring proof with keywords such as \lstinline|theorem|, \lstinline|lemma|. This is followed (optionally) by a unique name, as well as the fact statement, either in inner syntax
or in the more legible Isar style as above.

\pagebreak
\begin{lstlisting}
lemma no_empty_paths:
  assumes "Q\<in>\<P>"
  shows "Q\<noteq>{}"
proof -
  obtain a where "a\<in>\<E>"
    using nonempty_events by blast
  have "a\<in>Q \<or> a\<notin>Q" by auto
  thus ?thesis
  proof (rule disjE)
    assume "a\<in>Q"
    thus ?thesis by blast
  next
    assume "a\<notin>Q"
    then obtain b where "b\<in>\<emptyset> Q a"
      using two_in_unreach \<open>a\<in>\<E>\<close> assms
      by blast
    thus ?thesis
      using unreachable_subset_def by auto
  qed
qed
\end{lstlisting}


We start an Isar proof proof with the keyword \lstinline|proof|.
We can supplement \lstinline|proof| with an initial method to use (e.g. a case split \lstinline|rule disjE| as above or the general method \lstinline|safe|, which splits and rewrites goals; or \lstinline|induct| as explained below). Isabelle will try to choose a rule for us if we do not provide one, unless we prevent this using a dash (i.e. \lstinline|proof -|). A successful proof ends with \lstinline|qed|. Two other keywords can terminate a proof: \lstinline|sorry| and \lstinline|oops|. Both signify a proof that is not complete, or cannot be done, but while \lstinline|oops| means that Isabelle will refuse use of the unproven fact, \lstinline|sorry| allows an unproven statement to be used in legitimate proofs of other propositions. Thus \lstinline|sorry| can be quite dangerous (see Sec.~\ref{sec:bgr:isaloc} for an alternative). It is useful, however, for checking which subgoals could be sufficient to prove a lemma.

Intermediate facts are declared using for example \lstinline|have| or \lstinline|hence|, and facts that satisfy the current goal using \lstinline|show| or \lstinline|thus|. This is followed by an optional name and the fact statement, and proved using a separate \lstinline|proof (*...*) qed| block, with its own scope for variables. Such blocks can be nested. We may provide useful facts after \lstinline|using|, and a proof method or automatic theorem provers (ATP) after \lstinline|by|. Isabelle will now verify whether this method and collection of results are sufficient to prove the desired statement. The sixth line of the listing above is a simple example of this procedure.
Multiple facts can be listed after a single name, and proved all at once; such facts can be referenced by their given name, accompanied by a number in brackets that indicates which fact it was (e.g.\ \lstinline|factname(2)|).

Several results of our formalisation are proved by induction.
The method \lstinline|induct| takes an induction parameter, which is always of type \lstinline|nat| for our proofs, and splits the goal into subgoals, e.g. a base case and an inductive step.
Each subgoal is proved in its own scope, separated from the others by \lstinline|next|.
Isabelle provides shorthand notation for the usual first lines of both split cases. The \textit{base case} (\lstinline|case 0|) sets a goal that is just the lemma's conclusion, but with the induction variable set to $0$.
The induction case (\lstinline|case (Suc n)|) fixes $n$, assumes the lemma's conclusion for $n$, and sets the goal to the conclusion for $n+1$ (i.e. \lstinline|Suc n|). This assumption for $n$ is the induction hypothesis (IH).

Finally, several of our lemmas in Section~\ref{sec:order-path} use the keywords \lstinline|fixes|, which introduces a variable, and \lstinline|defines|, which gives its definition as an equality (strictly speaking, a meta-equality). Isabelle will treat the fixed variable as an abbreviation for its defining statement. We refer to the lemma \lstinline|show_segmentation| (part of Theorem~\ref{thm:11}) in Section~\ref{sec:order-path} as an example.

\subsubsection{Locales}\label{sec:bgr:isaloc}

One useful feature, particularly for sizeable axiom systems such as ours, is Isabelle's \lstinline|locale| mechanism. One can think of a locale as a parameterised context: it names one or more ``arbitrary but fixed'' parameters, and assumes some initial properties. In our case, these are undefined notions and axioms respectively. Since the formulation of axioms often changes as proofs are attempted because they are found wanting (e.g. axiom I6, see Theorem~\ref{thm:13} in Sec.~\ref{sec:unreach-connected}), we try to limit the amount of logic that is affected and possibly invalidated by such a change. Containing small groups of related axioms in their own separate locales circumscribes the scope of their influence. For instance, this purpose is served by our locale \lstinline|MinkowskiDense| (see again Sec.~\ref{sec:unreach-connected}, and below), which contains an assumption (in this case an additional, hidden assumption needed for one of Schutz' proofs) that we do not want to spill outside the locale. This is a safer alternative to using \lstinline|sorry|.

Locales have additional practical benefits: they are augmented by each theorem proven inside them, they can extend other locales, and they can be interpreted. The latter allows an explicit example to an abstract algebraic concept (e.g. $SO(3)$, 3D-rotations, form a concrete instance of a group).
This means that if we eventually want to find a model of our system, we can do so in steps: showing some interpretation $\mathcal{M}$ satisfies our locale \hbox{\lstinline|MinkowskiChain|} (see Sec.~\ref{sec:axioms:chains}) gives us immediate access to that locale's theorems (e.g. \lstinline|collinearity2|), and those of any locales it extends. These theorems may then be used to prove $\mathcal{M}$ satisfies the additional requirements of a locale extending \lstinline|MinkowskiChain|.

An example locale from our formalisation is given below. The locale \lstinline|MinkowskiDense| here extends \lstinline|MinkowskiSpacetime| with the additional assumption named \lstinline|path_dense|. The \emph{context block} of the locale is delimited by \lstinline|begin (*...*) end|.
Alternatively, the locale of an individual result can be specified directly using the keyword \lstinline|in| (fictitious example below).

\begin{lstlisting}
locale MinkowskiDense = MinkowskiSpacetime +
  assumes path_dense: "path ab a b \<Longrightarrow> \<exists>x. [[a x b]]"
begin
  lemma (in MinkowskiSpacetime) example: "True"
    by simp
end
\end{lstlisting}



\noindent
Since model proofs are outside the scope of this work, locales serve mostly an organisational purpose for our formalisation.

%% file: in/axioms.tex
Schutz proves several properties of his axiomatic system in his monograph \cite{schutz1997}:
consistency (relative to the real numbers), categoricity, and independence. He insists upon independence i.e.\ that none of his axioms can be derived from any combination of the others: he considers that the search for it has made his axioms more intuitive.


Some of the axioms as we encode them in Isabelle are subtly different from Schutz' statements.
These changes are due in some cases to the requirements of Isabelle/HOL (e.g.\ Isabelle's functions being total on types, not sets), in other cases some details are not considered in the original axiom, and several are just a matter of choice and simplicity (e.g.\ reformulations for easier use in interactive proofs). These choices will be discussed as we proceed with our exposition.
In most cases, Schutz' formulation can be easily restored as a theorem, by using the entire system of axioms.

Schutz lays out his axioms in two main groups: order and incidence. The first relates betweenness to events and paths, and establish a kind of plane geometry with axiom O6. The second deals with the relationships between events and paths, and also contains statements regarding unreachable subsets, which make a Euclidean/Galilean model impossible. In contrast to Schutz, we present axioms according to their specificity to Minkowski spacetime. In particular, our main comparison is with Hilbert's \textit{Grundlagen der Geometrie} \cite{hilbert1950}, which introduced the separation of incidence and order axioms.

\todo[fancyline]{A short discussion of sort against work that we've done mechanizing Hilbert will probably be needed somewhere. - Added something above, is it enough?}

Since several definitions of derived objects are required for stating some axioms, we construct our system as a hierarchy of locales (Sec.~\ref{sec:bgr:isaloc}), defining objects in the locale they make most sense in, and often just before they are needed.

\subsection{Primitives and Simple Axioms}\label{sec:axioms:primitives}

The first axioms, introduced in the locale \lstinline|MinkowskiPrimitive| together with the primitive notions of events and paths (which are introduced with the keyword \lstinline|fixes|), are similar to examples found in many other geometric axiom systems, notably Hilbert \cite{hilbert1950}. Schutz names them I1, I2, I3 \cite[p.~13]{schutz1997}, and they assert basic properties of two primitives: a set of events, $\E$, and a set of paths, $\P$, where each path is a set of events. 

\begin{theopargself}
\begin{axiom}[\customlabel{ax:I1}{I1} (Existence)]
$\E$ is not empty.
\end{axiom}
\begin{axiom}[\customlabel{ax:I2}{I2} (Connectedness)]
For any two distinct events $a,b \in \E$ there are paths $R$, $S$ such that $a \in R$, $b\in S$, and $R\cap S\neq \emptyset$.
\end{axiom}
\begin{axiom}[\customlabel{ax:I3}{I3} (Uniqueness)]
For any two distinct events, there is at most one path which contains both of them.
\end{axiom}
\end{theopargself}

As an example for the verbosity of a full formalisation, contrast Axiom \ref{ax:I3} with the many premises of its formalised version \lstinline|eq_paths|, and its customary translation of ``there is at most one'' as ``if given two such objects, they must be equal''.
Importantly, note that we also require one axiom Schutz does not have: \lstinline|in_path_event|, which excludes the possibility of non-event objects of the appropriate type being in a path, and guarantees $\P$ is a subset of the powerset of $\E$, not the universal set.

\begin{lstlisting}
locale MinkowskiPrimitive =
  fixes \<E> :: "'a set"
    and \<P> :: "('a set) set"
  assumes in_path_event [simp]: "\<lbrakk>Q \<in> \<P>; a \<in> Q\<rbrakk> \<Longrightarrow> a \<in> \<E>"
      (* I1 *)
      and nonempty_events [simp]: "\<E> \<noteq> {}"
      (* I2 *)
      and events_paths:
      	"\<lbrakk>a \<in> \<E>; b \<in> \<E>; a \<noteq> b\<rbrakk>
      	\<Longrightarrow> \<exists>R\<in>\<P>. \<exists>S\<in>\<P>. a \<in> R \<and> b \<in> S \<and> R \<inter> S \<noteq> {}"
      (* I3 *)
      and eq_paths [intro]:
        "\<lbrakk>P \<in> \<P>; Q \<in> \<P>; a \<in> P; b \<in> P; a \<in> Q; b \<in> Q; a\<noteq>b\<rbrakk> \<Longrightarrow> P = Q"
\end{lstlisting}

Nothing initially defines $\E$ apart from the type of its elements, yet we do not take $\E$ to be the universal set of type \lstinline|'a|.
This choice is made since it may lead to easier model instantiations in the future: for example, it allows building a model from a subset of natural numbers without defining an extra datatype. Given Isabelle's lack of subtypes, if events were the universal set of some type, a model over a subset of natural numbers could not make immediate recourse to the type \lstinline|nat|, but would need to define an entirely new type.
A universal set of events would also differ from Schutz' language. For example, types are never empty in Isabelle, so a universal set of events already implies Axiom~\ref{ax:I1}.
The set of paths $\P$ is always envisaged as a strict subset of the powerset of $\E$ -- otherwise the axioms introduced later in Sec.~\ref{sec:axioms:unreach} lose all relevance.

Our final undefined notion, the ternary relation of betweenness, is defined over events.
It is introduced in the locale \lstinline|MinkowskiBetweenness|, which extends \hbox{\lstinline|MinkowskiPrimitive|} and contains the first five axioms of order (O1 - O5) \cite[p.~10]{schutz1997}.

The axioms of order in Schutz' system are in close analogy with axioms of the same name in Hilbert's \emph{Grundlagen} (i.e.\ his group II). Hilbert's Axiom~II.1 combines Schutz' Axioms~\ref{ax:O1}, \ref{ax:O2}, \ref{ax:O3}; Hilbert's II.2 becomes Schutz' Theorem~\ref{thm:6}, II.3 becomes Theorem~\ref{thm:1}. Pasch's axiom exists in both systems, respectively as II.4 and \ref{ax:O6}.

\begin{theopargself}
\begin{axiom}[\customlabel{ax:O1}{O1}]
For events $a,b,c \in \E$,
\[\ord{a}{b}{c} \implies \exists Q \in \P : a,b,c \in Q.\]
\end{axiom}
\begin{axiom}[\customlabel{ax:O2}{O2}]
For events $a,b,c \in \E$,
\[\ord{a}{b}{c} \implies \ord{c}{b}{a}.\]
\end{axiom}
\begin{axiom}[\customlabel{ax:O3}{O3}]
For events $a,b,c \in \E$,
\[\ord{a}{b}{c} \implies a,b,c \text{ are distinct.}\]
\end{axiom}
\begin{axiom}[\customlabel{ax:O4}{O4}]
For distinct events $a,b,c,d \in \E$,
\[\ord{a}{b}{c} \text{ and } \ord{b}{c}{d} \implies \ord{a}{b}{d}\;.\]
\end{axiom}
\begin{axiom}[\customlabel{ax:O5}{O5}]
For any path $Q \in \P$ and any three distinct events $a,b,c \in Q$,
\[
    \ord{a}{b}{c} \;\text{ or }\; \ord{b}{c}{a} \;\text{ or }\; \ord{c}{a}{b} \;\text{ or } \\
    \ord{c}{b}{a} \;\text{ or }\; \ord{a}{c}{b} \;\text{ or }\; \ord{b}{a}{c} \;.
\]
\end{axiom}
\end{theopargself}

Schutz denotes betweenness as $\ord{\_}{\_}{\_}$, but since that notation is used for lists in Isabelle, we define it to be \lstinline|[[_ _ _]]| below.

\begin{lstlisting}
locale MinkowskiBetweenness = MinkowskiPrimitive +
  fixes betw :: "'a \<Rightarrow> 'a \<Rightarrow> 'a \<Rightarrow> bool" ("[[_ _ _]]")
      (* O1 *)
  assumes abc_ex_path: "[[a b c]] \<Longrightarrow> \<exists>Q\<in>\<P>. a \<in> Q \<and> b \<in> Q \<and> c \<in> Q"
      (* O2 *)
      and abc_sym: "[[a b c]] \<Longrightarrow> [[c b a]]"
      (* O3 *)
      and abc_ac_neq: "[[a b c]] \<Longrightarrow> a \<noteq> c"
      (* O4 *)
      and abc_bcd_abd: "\<lbrakk>[[a b c]]; [[b c d]]\<rbrakk> \<Longrightarrow> [[a b d]]"
      (* O5 *)
      and some_betw:
        "\<lbrakk>Q \<in> \<P>; a \<in> Q; b \<in> Q; c \<in> Q; a \<noteq> b; a \<noteq> c; b \<noteq> c\<rbrakk>
               \<Longrightarrow> [[a b c]] \<or> [[b c a]] \<or> [[c a b]]"
\end{lstlisting}

Three of these have mild changes compared to Schutz: O3 and O5 are slightly weaker (having weaker conclusions) since the original statements are actually derivable (in the same locale). In O4, Schutz' condition that $a,b,c,d$ be distinct has to be removed. This is because distinctness of $a,c$ and $b,d$ is already implied by \ref{ax:O3}, and requiring $a\neq d$ makes Schutz' proof of Theorem~\ref{thm:1} impossible (see Sec.~\ref{sec:order-fin-chain}).

We prove Schutz' Axiom~\ref{ax:O3} from our formulation of Axioms~\ref{ax:O2}, \ref{ax:O3}, \ref{ax:O4}; and Schutz' Axiom~\ref{ax:O5} from our \ref{ax:O2} and \ref{ax:O5}.

\subsection{Chains}\label{sec:axioms:chains}

The final axiom of order given by Schutz is analogous to the axiom of Pasch, which is common in axiomatic geometric systems. It is stated in terms of particular subsets of paths called \emph{chains}, which Schutz defines as follows \cite[p.~11]{schutz1997}.

\begin{definition}
A sequence of events $\;Q_0, Q_1, Q_2,\; \dots\;$ (of a path $Q$) is called a \emph{chain} if:
\begin{enumerate}
\item[(i)] it has two distinct events, or
\item[(ii)] it has more than two distinct events and for all $i \geq 2$, $$\ord{Q_{i-2}}{Q_{i-1}}{Q_{i}}\;.$$
\end{enumerate}
\end{definition}

This is hard to represent in Isabelle because of the notion of a sequence as an indexed set. The informal naming convention of using a label $Q_i$ for an event encodes two pieces of information: that the event lies on path $Q$, and that several betweenness relations hold with other events indexed by adjacent natural numbers.
Following Palmer and Fleuriot \cite{palmer2018} and Scott \cite[p.~110]{scott2015}, we explicitly give a function $I \rightarrow Q$ (with $I \subseteq \setN$) that is order-preserving, and use this to define chains.
The predicate \lstinline|ordering| formalises what we mean by ``order-preserving'', taking as arguments an indexing function \lstinline|f|, a ternary relation \lstinline|ord| on the codomain of \lstinline|f|, and a set of events \lstinline|X|.


\begin{lstlisting}
definition ordering ::
  "(nat \<Rightarrow> 'a) \<Rightarrow> ('a \<Rightarrow> 'a \<Rightarrow> 'a \<Rightarrow> bool) \<Rightarrow> 'a set \<Rightarrow> bool"
  where "ordering f ord X
    \<equiv> \<forall>n. (finite X \<longrightarrow> n < card X) \<longrightarrow> f n \<in> X \<and>
      \<forall>x\<in>X. (\<exists>n. (finite X \<longrightarrow> n < card X) \<and> f n = x) \<and>
      \<forall>n n' n''. (finite X \<longrightarrow> n'' < card X) \<and> n<n' \<and> n'<n''
    		\<longrightarrow> ord (f n) (f n') (f n'')"
\end{lstlisting}

Our chains differ from Schutz' in that they use sets instead of his sequences, and that while he assumes chains to lie on paths, we prove this as a theorem (\lstinline|chain_on_path|).
We also have a stronger condition on preserving long-range order: in our case, $\ord{f(n)}{f(n')}{f(n'')}$ must hold for any $n<n'<n''$, while Schutz only considers $n+1 = n' = n''-1$.%
\footnote{A kind of chain that is more precisely similar to Schutz' definition is briefly introduced in Sec.~\ref{sec:order-fin-chain}.}
Notice that we split the definition between chains of two events, \lstinline|short_ch|, and chains with at least three events, \lstinline|long_ch_by_ord|, as Schutz does. The abbreviation \lstinline|path_ex| used in the definition of the two-event chain asserts that two elements are distinct, and that there is a path containing both.
The cardinality of a set $X$, denoted $|X|$ in prose, is \lstinline|card X| in Isabelle. It is a natural number, and infinite sets have cardinality $0$, just like the empty set does. The conditions involving cardinality in \lstinline|ordering| are used to ensure that a natural number is a valid index into the chain.


\begin{lstlisting}
definition short_ch :: "'a set \<Rightarrow> bool"
  where "short_ch X \<equiv>
    \<exists>x\<in>X. \<exists>y\<in>X. path_ex x y \<and> \<not>(\<exists>z\<in>X. z\<noteq>x \<and> z\<noteq>y)"

definition long_ch_by_ord :: "(nat \<Rightarrow> 'a) \<Rightarrow> 'a set \<Rightarrow> bool"
  where "long_ch_by_ord f X \<equiv>
    \<exists>x\<in>X. \<exists>y\<in>X. \<exists>z\<in>X. x\<noteq>y \<and> y\<noteq>z \<and> x\<noteq>z \<and> ordering f betw X"
\end{lstlisting}
\begin{lstlisting}
definition fin_long_chain :: "(nat\<Rightarrow>'a)\<Rightarrow>'a\<Rightarrow>'a\<Rightarrow>'a\<Rightarrow>'a set\<Rightarrow>bool"
    ("[_[_ .. _ ..  _]_]")
  where "fin_long_chain f x y z Q \<equiv>
    x\<noteq>y \<and> x\<noteq>z \<and> y\<noteq>z \<and> finite Q \<and> long_ch_by_ord f Q \<and>
    f 0 = x \<and> y\<in>Q \<and> f (card Q - 1) = z"
\end{lstlisting}

\noindent
Two auxiliary definitions are made to capture Schutz' prose definitions more directly.
\begin{lstlisting}
definition ch_by_ord :: "(nat \<Rightarrow> 'a) \<Rightarrow> 'a set \<Rightarrow> bool"
  where "ch_by_ord f X \<equiv>
    short_ch X \<or> long_ch_by_ord f X"

definition ch :: "'a set \<Rightarrow> bool"
  where "ch X \<equiv> \<exists>f. ch_by_ord f X"
\end{lstlisting}

\noindent
We point out the notation: a \lstinline|fin_long_chain| is denoted \lstinline|[f[x..y..z]X]|, and we carry the indexing function and the set of all chain elements explicitly; this is absorbed into Schutz' subscripting notation.
We are now ready to describe the final axiom of order.

\begin{theopargself}
\begin{axiom}[\customlabel{ax:O6}{O6}]
If $Q$, $R$, $S$ are distinct paths which meet at events $a \in Q \cap R$, $b \in Q\cap S$, $c \in R \cap S$ and if:
\begin{enumerate}
    \item[(i)] there is an event $d \in S$ such that $\ord{b}{c}{d}$, and
    \item[(ii)] there is an event $e \in R$ and a path $T$ which passes through both $d$ and $e$ such that $\ord{c}{e}{a}$,
\end{enumerate}
then $T$ meets $Q$ in an event $f$ which belongs to a finite chain $\chain[\cdot\cdot]{a,f,b}$.
\end{axiom}
\end{theopargself}

\begin{lstlisting}
locale MinkowskiChain = MinkowskiBetweenness +
  assumes O6:
    "\<lbrakk>Q \<in> \<P>; R \<in> \<P>; S \<in> \<P>; T \<in> \<P>; Q \<noteq> R; Q \<noteq> S; R \<noteq> S;
      a \<in> Q\<inter>R \<and> b \<in> Q\<inter>S \<and> c \<in> R\<inter>S;
      \<exists>d\<in>S. [[b c d]] \<and> (\<exists>e\<in>R. d \<in> T \<and> e \<in> T \<and> [[c e a]])\<rbrakk>
    \<Longrightarrow> \<exists>f\<in>T\<inter>Q. \<exists>X. [[a..f..b]X]"
\end{lstlisting}

Although the statement is technical, the intention of O6 (or Pasch's axiom) is simple.
Using some intuition from Euclidean geometry, a rough translation is: if three paths meet in a triangle, then a fourth path which intersects one side of the triangle externally, and another internally, must meet the third side internally as well (see Fig.~\ref{fig:O6}).
Such an intuitive understanding can be justified by noting that similar axioms occur e.g. in Hilbert's \textit{Grundlagen} \cite{hilbert1950} and its mechanisation \cite{meikle2003}; it is not O6 that makes our system non-Euclidean.

\enlargethispage{5\baselineskip}
\begin{figure}[ht!]
\centering
\includegraphics[width=.5\textwidth]{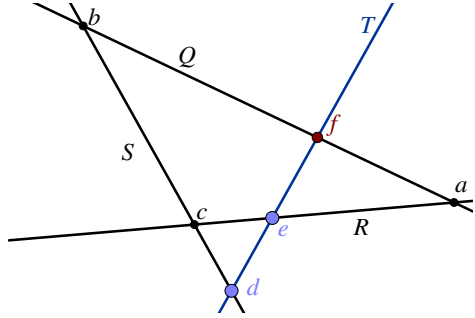}
\caption{\label{fig:O6}Intuitive visualisation of axiom O6. A path $T$ that meets $S$ externally to the triangle $QRS$ (in $d$) and meets $R$ internally (in $e$), must meet the third side of the triangle internally (in $f$).}
\end{figure}

\subsection{Unreachability}\label{sec:axioms:unreach}
While the axioms of the previous sections establish a geometry, nothing in them excludes a Euclidean space with Galilean relativity, i.e. velocities that are additive across reference frames \cite[p.~12]{schutz1997}.
Crucially, no speed limit is implied so far, and thus there is no trajectory through space and time that is forbidden.
The next group of axioms (I5-I7) specifies existence and basic properties of unreachable sets, a concept tightly linked to the lightcones often used in relativistic physics \cite[sec.~1.4]{gourgoulhon2013g}. In fact, if we pre-empt significantly, and hypothesise our undefined paths to relate to observer worldlines, one can glean the notion of an ultimate speed limit hidden in the condition that certain regions of spacetime should not be connected by paths. Ultimately, saying that nothing can move faster than some speed $c$ is merely the statement that certain histories or trajectories through space and time should not occur. We begin by formalising Schutz' various notions of unreachable sets.

\begin{definition}[Unreachable Subset from an Event] 
Given a path $Q$ and an event $b \notin Q$, we define the unreachable subset of $Q$ from $b$ to be
$$Q(b,\emptyset) := \left\lbrace x : \text{there is no path which contains $b$ and $x$}, x \in Q \right\rbrace \;.$$
\end{definition}
\begin{lstlisting}
definition unreachable_subset ::
  "'a set \<Rightarrow> 'a \<Rightarrow> 'a set" ("\<emptyset> _ _" [100, 100] 100)
  where "unreachable_subset Q b
           \<equiv> {x\<in>Q. Q \<in> \<P> \<and> b \<in> \<E> \<and> b \<notin> Q \<and> \<not>(path_ex b x)}"
\end{lstlisting}

The pen-and-paper definition is simple enough: it collects all the events $x$ of a path $Q$ that cannot be connected (by a path) to another event $b \notin Q$.
In prose, we use Schutz' notation $Q(b,\emptyset)$, where $\emptyset$ is used like a flag filtering elements of $Q$ whereas the Isabelle version uses \lstinline|\<emptyset> Q b|, where $\emptyset$ behaves as a function symbol. Note that the empty set in Isabelle is denoted \lstinline|{}|, so ambiguity is not an issue.

The second definition is more complex: if $Q$ meets $R$ at $x$, then we use the notation \linebreak\lstinline|\<emptyset> Q from Qa via R at x| to collect all events $Q_y \in Q$ that are on the side of the intersection $x$ given by $Q_a$, and where some event on $R$ is connected neither to $Q_a$ nor $Q_y$ (see Fig.~\ref{fig:unreach_via}).

\begin{definition}[Unreachable Subset via a Path~{\cite[pp.~16]{schutz1997}}]
For any two distinct paths $Q$, $R$ which meet at an event $x$, we define the unreachable subset of $Q$ from $Q_a$ via $R$ to be
$$Q(Q_a,R,x,\emptyset) := \left\lbrace Q_y : \ord{x}{Q_y}{Q_a} \text{ and } \exists R_w \text{ such that } Q_a, Q_y \in Q(R_w,\emptyset) \right\rbrace \;.$$
\end{definition}
\begin{lstlisting}
definition unreachable_subset_via ::
  "'a set \<Rightarrow> 'a \<Rightarrow> 'a set \<Rightarrow> 'a \<Rightarrow> 'a set"
    ("\<emptyset> _ from _ via _ at _" [100, 100, 100, 100] 100)
  where "unreachable_subset_via Q Qa R x
    \<equiv> {Qy. [[x Qy Qa]] \<and> (\<exists>Rw\<in>R. Qa \<in> \<emptyset> Q Rw \<and> Qy \<in> \<emptyset> Q Rw)}"
\end{lstlisting}

\begin{figure}
\centering
\includegraphics[width=.5\textwidth]{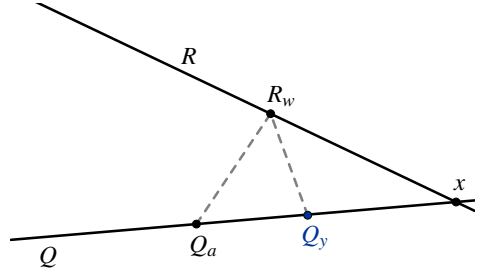}
\caption{\label{fig:unreach_via}The event $Q_y$ belongs to the unreachable subset of $Q$ from $Q_a$ via $R$. Thus there is an event $R_w$, such that there are no paths connecting $(Q_a, R_w)$ or $(Q_y, R_w)$ (dashed lines). In this case, $R_w$ also belongs to the unreachable subset of $R$ from $Q_a$.}
\end{figure}

Next, we give the formalised axioms I5-I7, introduced in the locale \linebreak\hbox{\lstinline|MinkowskiUnreachable|}, together with their prose formulation and some comment.
Axiom \ref{ax:I5} is simple once unreachable sets from events are understood. It has important implications for many proofs, since it is necessary to guarantee that the empty set is not a path (see Sec.~\ref{sec:bgr:isaind}, where this result serves as an example listing). It is the only axiom that mentions the existence of events on a path.

\begin{theopargself}
\begin{axiom}[\customlabel{ax:I5}{I5}]
For any path $Q$ and any event $b \notin Q$, the unreachable set $Q(b,\emptyset)$ contains (at least) two events.
\end{axiom}
\end{theopargself}
\todo[fancyline]{check globally and in Isa: chains are XYZ, paths are PQR}
\begin{lstlisting}
locale MinkowskiUnreachable = MinkowskiChain +
  assumes (*I5*) two_in_unreach:
    "\<lbrakk>Q \<in> \<P>; b \<in> \<E>; b \<notin> Q\<rbrakk> \<Longrightarrow> \<exists>x\<in>\<emptyset> Q b. \<exists>y\<in>\<emptyset> Q b. x \<noteq> y"
\end{lstlisting}

Schutz calls axiom I6 ``Connectedness of the Unreachable Set''. Indeed, given two unreachable (from $b$) events $Q_x, Q_z$ on a path $Q$, it essentially states that any points between $Q_x, Q_z$ must be unreachable too. This is phrased in terms of a finite chain with endpoints $Q_x,Q_z$.
\begin{theopargself}
\begin{axiom}[\customlabel{ax:I6}{I6}]
Given any path $Q$, any event $b \notin Q$ and distinct events $Q_x, Q_z \in Q(b,\emptyset)$, there is a finite chain $[Q_0 \;\dots\; Q_n]$ with $Q_0 = Q_x$ and $Q_n = Q_z$ such that for all $i \in \left\lbrace1,2,\dots,n\right\rbrace$,
\begin{enumerate}
    \item[(i)] $Q_i \in Q(b,\emptyset)$
    \item[(ii)] $\ord{Q_{i-1}}{Q_y}{Q_i} \implies Q_y \in Q(b,\emptyset)$.
\end{enumerate}
\end{axiom}
\end{theopargself}

Notice the extra clause for short chains in the formalisation: if we have only two events, ternary ordering is meaningless, thus so is $f$. This means that while Schutz often just doesn't mention two-event chains (supposing perhaps that this part of a proof is obvious), Isar statements and proofs have to be split, making them more complicated. The two-event clause was needed for the proof of Theorem~\ref{thm:13} (see Sec~\ref{sec:unreach-connected}).
\begin{lstlisting}
    assumes I6:
    "\<lbrakk>Q \<in> \<P>; b \<notin> Q; b \<in> \<E>; Qx \<in> (\<emptyset> Q b); Qz \<in> (\<emptyset> Q b)\<rbrakk>
     \<Longrightarrow> \<exists>X f. ch_by_ord f X \<and> f 0 = Qx \<and> f (card X - 1) = Qz \<and>
         (\<forall>i\<in>{1 .. card X - 1}. (f i) \<in> \<emptyset> Q b \<and>
             (\<forall>Qy\<in>\<E>. [[(f(i-1)) Qy (f i)]] \<longrightarrow> Qy \<in> \<emptyset> Q b)) \<and>
         (short_ch X \<longrightarrow> Qx\<in>X \<and> Qz\<in>X \<and>
            (\<forall>Qy\<in>\<E>. [[Qx Qy Qz]] \<longrightarrow> Qy \<in> \<emptyset> Q b))"
\end{lstlisting}


Axiom I7 about the ``Boundedness of the Unreachable Set'' is reminiscent of the Archi-medean property\todo{is this a good idea? archimedean is about fixed-size steps, while we have no length yet}, namely that one can ``leave'' the unreachable set in finitely many ``steps''. A simplified illustration is given in Fig.~\ref{fig:thm4}.
\begin{theopargself}
\begin{axiom}[\customlabel{ax:I7}{I7}]
Given any path $Q$, any event $b \notin Q$, and events $Q_x \in Q \setminus Q(b,\emptyset)$ and $Q_y \in Q(b,\emptyset)$, there is a finite chain
$$\chain{Q_0,Q_m,Q_n}$$ 
with $Q_0 = Q_x$, $Q_m = Q_y$ and $Q_n \in Q \setminus Q(b,\emptyset)$.
\end{axiom}
\end{theopargself}
We drop the double naming of the events $Q_x=Q_0$ and $Q_y=Q_m$, noting the index of $Q_x$ is implied once the chain $\chain{Q_x,Q_y,Q_n}$ is defined. The complement of the unreachable set, $Q \setminus Q(b,\emptyset)$, is best thought of as all the events of path $Q$ that \emph{can} be reached by a path passing through $b$. Axiom I7 is then straightforwardly formalised as:
\begin{lstlisting}
    assumes I7:
    "\<lbrakk>Q \<in> \<P>; b \<notin> Q; b \<in> \<E>; Qx \<in> Q - \<emptyset> Q b; Qy \<in> \<emptyset> Q b\<rbrakk>
     \<Longrightarrow> \<exists>g X Qn. [g[Qx..Qy..Qn]X] \<and> Qn \<in> Q - \<emptyset> Q b"
\end{lstlisting}

\subsection{Symmetry and Continuity}\label{sec:axioms:sym-cont}

The final two axioms, symmetry and continuity, both receive their own locale.
Although neither is used in proofs in this paper, we still present them in full as they are non-trivial to formalise in Isabelle.

The axiom of symmetry is a hefty statement that, according to Schutz \cite{schutz1997}, serves as a replacement of an entire axiom group in geometries such as Hilbert's \textit{Grundlagen}. Continuity is simple to state, but relies on mechanised definitions of bounds and closest bounds. 
We break up the presentation of the formalised axiom of symmetry, explaining the conclusion as we go along. See also Figure~\ref{fig:symmetry-axiom}.

\begin{theopargself}
\begin{axiom}[\customlabel{ax:symmetry}{S (Symmetry)}]
If $Q,R,S$ are distinct paths which meet at some event $x$ and if $Q_a \in Q$ is an event distinct from $x$ such that
$$Q(Q_a,R,x,\emptyset) = Q(Q_a,S,x,\emptyset)$$
then
\begin{enumerate}
    \item[(i)]
there is a mapping $\theta:\E\longrightarrow\E$
    \item[(ii)]
which induces a bijection $\Theta:\P \longrightarrow \P$, such that
    \item[(iii)]
the events of $Q$ are invariant, and
    \item[(iv)]
$\Theta : R \longrightarrow S$.
\end{enumerate}
\end{axiom}
\end{theopargself}

\begin{figure}
    \centering
    \includegraphics[width=.5\textwidth]{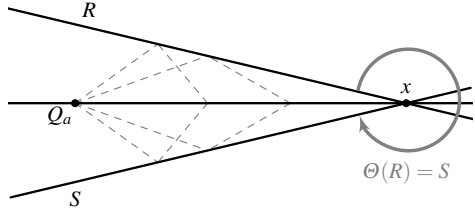}
    \caption{%
        \label{fig:symmetry-axiom}%
        Visualisation of Axiom~\ref{ax:symmetry}. The unreachable subsets of $Q$ from $Q_a$ via $R$ and $S$ (indicated by dashed lines) are equal, so the induced symmetry mapping $\Theta$ takes $R$ to $S$.
    }
\end{figure}

\begin{lstlisting}
locale MinkowskiSymmetry = MinkowskiUnreachable +
  assumes Symmetry:
    "\<lbrakk>Q \<in> \<P>; R \<in> \<P>; S \<in> \<P>; Q \<noteq> R; Q \<noteq> S; R \<noteq> S;
    x \<in> Q\<inter>R\<inter>S; Q\<^sub>a \<in> Q; Q\<^sub>a \<noteq> x;
    \<emptyset> Q from Q\<^sub>a via R at x = \<emptyset> Q from Q\<^sub>a via S at x\<rbrakk>
\end{lstlisting}

The first two lines essentially say that $Q,R,S$ are distinct paths in \lstinline|SPRAY x| (see Sec.~\ref{sec:axioms:dim}), and obtain an event $Q_a \neq x$ on $Q$.
The third states that the unreachable sets of $Q$ from the source $x$ via $R$ and $S$ are the same.

\todo[fancyline]{commented paragraph: rewrite?}

We split up the conclusion of the axiom, reproducing Schutz' prose \cite[p.~16]{schutz1997} for each of the parts (i)-(iv); notice the first line below quantifies the entire conclusion.\\
\begin{enumerate}
    \item[(i)]
there is a mapping $\theta:\E\longrightarrow\E$
\begin{lstlisting}
    \<Longrightarrow> \<exists>\<theta>::'a\<Rightarrow>'a.
\end{lstlisting}
    \item[(ii)]
which induces%
\footnote{Schutz doesn't give an explicit form for $\Theta$. Since the set of paths is contained in the powerset of events, taking the direct image under $\theta$ to be the induced bijection seems the only choice.}
a bijection $\Theta:\P \longrightarrow \P$
\begin{lstlisting}
                bij_betw (\<lambda>P. {\<theta> y | y. y\<in>P}) \<P> \<P> \<and>
\end{lstlisting}
    \item[(iii)]
the events of $Q$ are invariant, and
\begin{lstlisting}
                (y\<in>Q \<longrightarrow> \<theta> y = y) \<and> 
\end{lstlisting}
    \item[(iv)]
$\Theta : R \longrightarrow S$
\begin{lstlisting}
                (\<lambda>P. {\<theta> y | y. y\<in>P}) R = S
\end{lstlisting}
\end{enumerate}

Schutz' statement is not completely clear on whether he means $Q$ to be invariant under $\theta$ or $\Theta$. We settled on the stronger version, involving $\theta$-invariance: it is stronger than the alternative only by also preserving the ordering of the events on $Q$. Since this ordering affects unreachable sets, not preserving it seemed to go against the spirit of the axiom.

The axiom of continuity compares to the property of least upper bounds on the real numbers (also called Dedekind completeness). Indeed, Theorem~12 (entitled ``Continuity''), the first to use this axiom, deals with sets that look very similar to Dedekind cuts \cite{dedekind1963}.
Bounds are defined by Schutz only for infinite chains.


\begin{definition}[(Closest) Bound {\cite[pp. 17]{schutz1997}}]
Given a path $Q \in \P$ and an infinite chain $[Q_0, Q_1 \;\dots\; ]$ of events in $Q$, the set
$$\mathcal{B} = \left\lbrace Q_b : i < j \implies \ord{Q_i}{Q_j}{Q_b}; Q_i, Q_j, Q_b \in Q\right\rbrace \;,$$
is called the \emph{set of bounds} of the chain: if $\mathcal{B}$ is non-empty we say that the chain is \emph{bounded}. If there is a bound $Q_b \in \mathcal{B}$ such that for all $Q_{b'} \in \mathcal{B} \setminus \left\lbrace Q_b \right\rbrace$,
$$\ord{Q_0}{Q_b}{Q_{b'}}$$
we say that $Q_b$ is a \emph{closest bound}.
\end{definition}

\begin{theopargself}
\begin{axiom}[\customlabel{ax:C}{C} (Continuity)]
Any bounded infinite chain has a closest bound.
\end{axiom}
\end{theopargself}

The formalisation in this case is straightforward. We formally define bounds first.
\begin{lstlisting}
definition is_bound_f :: "'a \<Rightarrow> 'a set \<Rightarrow> (nat\<Rightarrow>'a) \<Rightarrow> bool" where
  "is_bound_f Q\<^sub>b Q f \<equiv>
    \<forall>i j ::nat. [f[(f 0)..]Q] \<and> (i<j \<longrightarrow> [[(f i) (f j) Q\<^sub>b]])"
definition bounded :: "'a set \<Rightarrow> bool" where
  "bounded Q \<equiv> \<exists> Q\<^sub>b f. is_bound_f Q\<^sub>b Q f"
definition closest_bound :: "'a \<Rightarrow> 'a set \<Rightarrow> bool" where
  "closest_bound Q\<^sub>b Q \<equiv> \<exists>f. is_bound_f Q\<^sub>b Q f \<and>
    (\<forall> Q\<^sub>b'. (is_bound Q\<^sub>b' Q \<and> Q\<^sub>b' \<noteq> Q\<^sub>b) \<longrightarrow> [[(f 0) Q\<^sub>b Q\<^sub>b']])"
\end{lstlisting}

\noindent
The axiom of continuity is now so simple that the Isabelle locale below is easily readable.
\begin{lstlisting}
locale MinkowskiContinuity = MinkowskiSymmetry +
  assumes Continuity: "bounded Q \<longrightarrow> (\<exists>Q\<^sub>b. closest_bound Q\<^sub>b Q)"
\end{lstlisting}

\subsection{Path Dependence and Dimension}\label{sec:axioms:dim}
The final axiom we introduce is that of dimension. It comes last in our hierarchy of locales because spacetimes in different numbers of dimensions can then be constructed. Thus we found it sensible to have an easily replaceable top layer that specifies only the axiom least critical to the general Minkowski spacetime structure, in case one wants to explore other dimensions.

However, this axiom has a hidden purpose much more fundamental than we first realised: it is the only one that excludes a singleton set of events with an empty set of paths from being a model. As a result, the axiom of dimension turns out to be crucial to several fairly basic proofs involving geometric construction of several paths (that without it could not be guaranteed to exist), and we end up working inside the full \lstinline|MinkowskiSpacetime| locale for many more proofs than originally expected (notably, any proof requiring the overlapping ordering lemmas presented in Sec.~\ref{sec:order-path}). A minor restructuring could isolate an axiom for existence of at least one path: if applications in higher or lower dimensions are deemed important in future work, this is easily done.%
\footnote{This may not even break independence, as Schutz' independence model for I4 is simply 1+1-dimensional spacetime.}
We keep Schutz' formulation for now.

Defining dimensionality in linear algebra requires the idea of linear dependence and independence.
Since vector spaces are not included in our axioms, we need a more basic notion, namely an idea of paths depending on other paths.
This relation is defined only for a set of paths that all cross in one point and is called a $\spray$ \cite[p.~13]{schutz1997}.
\begin{definition}
Given any event $x$,
$$\spray[x] := \left\lbrace R: R \ni x, R \in \P\right\rbrace \;.$$
\end{definition}
\begin{lstlisting}
definition SPRAY :: "'a \<Rightarrow> ('a set) set"
  where "SPRAY x \<equiv> {R\<in>\<P>. x \<in> R}"
\end{lstlisting}

\noindent Path dependence in a $\spray$ is defined first for a set of three paths \cite[p.~13]{schutz1997}:

\begin{definition}
A subset of three paths of a $\spray$ is \textit{dependent} if there is a path which does not belong to the $\spray$ and which contains one event from each of the three paths: we also say any one of the three paths is \textit{dependent on} the other two. Otherwise the subset is \textit{independent}.
\end{definition}

\begin{lstlisting}
definition dep3_event :: "'a set \<Rightarrow> 'a set \<Rightarrow> 'a set \<Rightarrow> 'a \<Rightarrow> bool"
  where "dep3_event Q R S x
    \<equiv> Q \<noteq> R \<and> Q \<noteq> S \<and> R \<noteq> S
      \<and> Q \<in> SPRAY x \<and> R \<in> SPRAY x \<and> S \<in> SPRAY x
      \<and> (\<exists>T\<in>\<P>. T \<notin> SPRAY x
        \<and> (\<exists>y\<in>Q. y \<in> T) \<and> (\<exists>y\<in>R. y \<in> T) \<and> (\<exists>y\<in>S. y \<in> T))"
\end{lstlisting}

To obtain path dependence for an arbitrary number of paths, we extend the base case above by induction, quoting Schutz \cite[p.~14]{schutz1997}:

\begin{definition}
A path $T$ is \textit{dependent on} the set of $n$ paths (where $n \ge 3$)
$$S = \left\lbrace Q^{(i)} : i = 1, 2, \dots, n;\; Q^{(i)} \in \spray[x]\right\rbrace $$
if it is dependent on two paths $S^{(1)}$ and $S^{(2)}$, where each of these two paths is dependent on some subset of $n - 1$ paths from the set $S$. We also say that the set of $n+1$ paths $S\cup \left\lbrace T\right\rbrace$ is a \textit{dependent set}. If a set of paths has no dependent subset, we say that the set of paths is an \textit{independent set}.
\end{definition}

\begin{lstlisting}
inductive dep_path :: "'a set \<Rightarrow> ('a set) set \<Rightarrow> 'a \<Rightarrow> bool"
  where
    dep_two: "dep3_event T A B x \<Longrightarrow> dep_path T {A, B} x"
  | dep_n: "\<lbrakk>S \<subseteq> SPRAY x; card S \<ge> 3; dep_path T {S1, S2} x;
      S' \<subseteq> S; S'' \<subseteq> S; card S' = card S - 1; card S'' = card S - 1;
      dep_path S1 S' x; dep_path S2 S'' x\<rbrakk>
        \<Longrightarrow> dep_path T S x"
\end{lstlisting}

This definition uses the keyword \lstinline|inductive|, which allows us to give a non-recursive base case and induction rules, to create the minimal set of triplets $T,S,x$ such that \linebreak\lstinline|dep_path T S x|.
Notice that we keep track of the (source of the) $\spray$ that the paths exist in explicitly, while Schutz keeps this implicit, referring to it as and when needed.
This leaves us with only the job of transforming this inductive definition into an analytical one, such that a set of paths can be examined and found dependent or not, rather than being able only to construct such sets to measure.

\begin{lstlisting}
definition dep_set :: "('a set) set \<Rightarrow> bool"
  where "dep_set S \<equiv> \<exists>x. \<exists>S'\<subseteq>S. \<exists>P\<in>(S-S'). dep_path P S' x"
  
definition indep_set :: "('a set) set \<Rightarrow> bool"
  where "indep_set S \<equiv> \<not>(\<exists>T \<subseteq> S. dep_set T)"
\end{lstlisting}

\noindent
Now the axiom of dimension can be given as follows, with a final definition:

\begin{definition}
A $\spray$ is a 3-$\spray$ if:
\begin{enumerate}
\item[(i)] it contains four independent paths, and
\item[(ii)] all paths of the $\spray$ are dependent on these four paths.
\end{enumerate}
\end{definition}
\begin{theopargself}
\begin{axiom}[\customlabel{ax:I4}{I4} (Dimension)]
If $\E$ is non-empty, then there is at least one 3-$\spray$.
\end{axiom}
\end{theopargself}

Notice Schutz introduces the Axiom~\ref{ax:I1} into the antecedent of Axiom~\ref{ax:I4}. This serves the purpose of conserving independence: the empty set is an obvious model for proving independence of \ref{ax:I1}, and in this current formulation, the empty event-set vacuously satisfies Axiom~\ref{ax:I4}.

Formalising the 3-$\spray$ in Isabelle/HOL is long because we need to introduce the four distinct paths, all of them in a $\spray$. The final two lines of the definition are the interesting ones. Much like the Axiom of Continuity, Dimension becomes very simple, even in Isabelle, once all the preparation is complete.
\begin{lstlisting}
definition three_SPRAY :: "'a \<Rightarrow> bool" where
  "three_SPRAY x \<equiv> \<exists>S1\<in>\<P>. \<exists>S2\<in>\<P>. \<exists>S3\<in>\<P>. \<exists>S4\<in>\<P>.
    S1 \<noteq> S2 \<and> S1 \<noteq> S3 \<and> S1 \<noteq> S4 \<and> S2 \<noteq> S3 \<and> S2 \<noteq> S4 \<and> S3 \<noteq> S4
    \<and> S1 \<in> SPRAY x \<and> S2 \<in> SPRAY x \<and> S3 \<in> SPRAY x \<and> S4 \<in> SPRAY x
    \<and> (indep_set {S1, S2, S3, S4})
    \<and> (\<forall>S\<in>SPRAY x. dep_path S {S1,S2,S3,S4} x)"

locale MinkowskiSpacetime = MinkowskiContinuity +
  (* I4 *)
  assumes ex_3SPRAY: "\<E> \<noteq> {} \<Longrightarrow> \<exists>x\<in>\<E>. three_SPRAY x"
\end{lstlisting}

%% file: in/thms-3.1.tex
\begin{theorem}\label{thm:1}
    If $\ord{a}{b}{c}$ then $\ord{c}{b}{a}$ and no other order.
\end{theorem}

The point of this theorem is really to exclude other orders, as $\ord{c}{b}{a}$ is explicitly established by Axiom~\ref{ax:O2}. Schutz proceeds by contradiction, and following him forced us to change Axiom~\ref{ax:O4}. For example, Schutz claims that $\ord{b}{c}{a}$ implies (with $\ord{a}{b}{c}$) the order $\ord{a}{b}{a}$ via Axiom~\ref{ax:O4}. This works only if Axiom~\ref{ax:O4} is changed to allow, in the notation of its definition in Sec.~\ref{sec:axioms}, the case $a=d$.
We obtain a contradiction from $\ord{a}{b}{a}$ and Axiom~\ref{ax:O3}, which applies here to give $a \neq a$.

\begin{lstlisting}
theorem theorem1:
  assumes abc: "[[a b c]]"
  shows "[[c b a]] \<and> \<not> [[b c a]] \<and> \<not> [[c a b]]"
\end{lstlisting}

Our formalisation is concerned only with two of the four impossible orderings, the rest being trivial via Axiom~\ref{ax:O2}.
In addition to \lstinline|theorem1|, we prove a similar result called \lstinline|abc_only_cba|. This concludes only the impossible orderings from $\ord{a}{b}{c}$, and is used frequently in the rest of the formalisation. It follows from \ref{ax:O2}, \ref{ax:O3}, and \ref{ax:O4} like Theorem~\ref{thm:1}.

\begin{lstlisting}
lemma abc_only_cba: 
  "[[a b c]] 
   \<Longrightarrow> \<not> [[b a c]] \<and> \<not> [[a c b]] \<and> \<not> [[b c a]] \<and> \<not> [[c a b]]"
\end{lstlisting}

The second theorem, ``Order on a Finite Chain'', begins building a link between Schutz' definition of chains, and ours (Sec.~\ref{sec:axioms:chains}) \cite[p.~18]{schutz1997}.
In fact, it allows to transform a chain with only local orderings (orderings of elements with adjacent indices) into one where any three events on the chain can be ordered -- the latter being true of our chains by definition. In this way, Theorem~\ref{thm:2} justifies our definition, since Schutz' chains can be immediately transformed into this stronger variety.

\begin{theorem}[Order on a Finite Chain]\label{thm:2}
    On any finite chain \chain{Q_0, Q_n}, there is a betweenness relation for each ordered triple; that is
    $$0 \leq i < j < l \leq n \implies \ord{Q_i}{Q_j}{Q_l}\;.$$
    Furthermore all events of a chain are distinct.
\end{theorem}

\noindent
This theorem is true by definition for the chains we define in Sec.~\ref{sec:axioms:chains}. Indeed, it can be verified by the prover \lstinline|metis| in a single line \cite{hurd2003,smolka2013}.

\begin{lstlisting}
theorem (*2*) order_finite_chain:
  assumes chX: "long_ch_by_ord f X"
      and finiteX: "finite X"
      and ordered_nats: "0 \<le> (i::nat) \<and> i < j \<and> j < l \<and> l < card X"
    shows "[[(f i) (f j) (f l)]]"
  by (metis ordering_def chX long_ch_by_ord_def ordered_nats)
\end{lstlisting}

In order to check that Schutz' proof holds,
we introduce a new definition for chains, \lstinline|long_ch_by_ord2|. This is closer to Schutz' original definition, and similar to \linebreak\hbox{\lstinline|long_ch_by_ord|},
except for imposing ordering relations only on adjacent events.

\begin{lstlisting}
definition ordering2 ::
  "(nat \<Rightarrow> 'a) \<Rightarrow> ('a \<Rightarrow> 'a \<Rightarrow> 'a \<Rightarrow> bool) \<Rightarrow> 'a set \<Rightarrow> bool"
  where "ordering2 f ord X
    \<equiv> (*...*) (\<forall>n n' n''.
      (finite X\<longrightarrow>n'' < card X) \<and> Suc n = n' \<and> Suc n' = n''
        \<longrightarrow> ord (f n) (f n') (f n''))"
\end{lstlisting}

\begin{lstlisting}
definition long_ch_by_ord2 ::
  "(nat \<Rightarrow> 'a) \<Rightarrow> 'a set \<Rightarrow> bool"
  where "long_ch_by_ord2 f X
    \<equiv> \<exists>x\<in>X. \<exists>y\<in>X. \<exists>z\<in>X. x\<noteq>y \<and> y\<noteq>z \<and> x\<noteq>z \<and> ordering2 f betw X"
\end{lstlisting}

We can then state Theorem~\ref{thm:2} using this new chain. Notice that Theorem~\ref{thm:2} strengthens the ordering relations between chain elements to an extent that is sufficient to prove equivalence between \lstinline|long_ch_by_ord| and \lstinline|long_ch_by_ord2|, provided the chains are finite. This is why we use the former in most of our formalisation: it gives immediate access to a more powerful relationship between chain events.


\begin{lstlisting}
theorem order_finite_chain2:
  assumes chX: "long_ch_by_ord2 f X"
      and finiteX: "finite X"
      and ordered_nats: "0 \<le> (i::nat) \<and> i < j \<and> j < l \<and> l < card X"
  shows "[[(f i) (f j) (f l)]]"
\end{lstlisting}

The proof of Theorem~\ref{thm:2} follows the outline of Schutz \cite[p.~19]{schutz1997}: it is split into two proofs by induction on decreasing $j$ for $j<l$, and increasing $j$ for $i<j$. The induction step propagates ordering relations along increasing/decreasing indices using Axioms~\ref{ax:O2} and \ref{ax:O4}.


Distinctness of chain events is an obvious conclusion of the first part of the theorem and Axiom~\ref{ax:O3}. Our explicit handling of indices allows for a clearer statement of this property, namely that distinct indices label distinct events (i.e. the indexing function is injective). Several such statements are included in the formalisation, and we give an example below. The proof relies notably on Axiom~\ref{ax:O3} only, but involves a few case splits according to how we can find a third element for the betweenness relation (e.g. whether a natural number exists between $i$ and $j$ or not).

\begin{lstlisting}
theorem (*2ii*) index_injective:
  fixes i::nat and j::nat
  assumes chX: "long_ch_by_ord2 f X"
      and finiteX: "finite X"
      and indices: "i<j" "j<card X"
    shows "f i \<noteq> f j"
\end{lstlisting}

Schutz follows the statement of Theorem~\ref{thm:2} with the remark that Theorem~\ref{thm:10} extends it to any finite subset of a path. Indeed, there is a tight relationship between these two results, and we will mention Theorem~\ref{thm:2} again in Sec.~\ref{sec:order-path}.

We can now prove an explicit claim of our chains being the same (in the finite case) as Schutz'. The proofs for each individual direction of the equivalence go through easily using Theorem~\ref{thm:2}.

\begin{lstlisting}
lemma ch_equiv:
  assumes "finite X"
  shows "long_ch_by_ord f X \<longleftrightarrow> long_ch_by_ord2 f X"
\end{lstlisting}


%% file: in/thms-3.2.tex
We begin by defining a fundamental structure for the geometric proofs to come. This can be intuitively thought of as a triangle -- while maintaining the reassurance that Isabelle will not allow us to use any unproven Euclidean intuition about triangles.

\begin{definition}[Kinematic Triangle]
A set of three distinct events $\left\lbrace a,b,c \right\rbrace$ is called a \textit{kinematic triangle} if each pair of events belongs to one of three distinct paths: we will refer to the kinematic triangle $\triangle abc$, or simply $\triangle abc$.
\end{definition}

Furthermore, since each path is defined by any two distinct points that lie on it (thanks to Axiom ~\ref{ax:I3}), we shall denote a path that contains two distinct events $a$ and $b$ as $ab$. In Isabelle, this shorthand is not possible, but we approximate it using the following Isabelle abbreviations. 

\begin{lstlisting}
abbreviation path :: "'a set \<Rightarrow> 'a \<Rightarrow> 'a \<Rightarrow> bool" where
  "path ab a b \<equiv> ab \<in> \<P> \<and> a \<in> ab \<and> b \<in> ab \<and> a \<noteq> b"

abbreviation path_of :: "'a \<Rightarrow> 'a \<Rightarrow> 'a set" where
  "path_of a b \<equiv> THE ab. path ab a b"
\end{lstlisting}

Theorem~3 is a straightforward application of the Axiom of Collinearity (\ref{ax:O6}, see also Fig.~\ref{fig:O6}), and named after it. Schutz provides three results of this name, of increasing complexity, with Theorem~7 being the other one included in our formalisation. The Third Collinearity Theorem, numbered 15, is fundamental to Schutz' treatment of optical lines and causality \cite[chap.~4]{schutz1997}. Its proof relies heavily on the preceding Collinearity Theorems.

\begin{theorem}[Collinearity]\label{thm:3}
    Given a kinematic triangle $\triangle abc$ and events $d,e$ such that
    \begin{enumerate}
        \item[(i)] there is a path $de$, and
        \item[(i)] $\ord{b}{c}{d}$ and $\ord{c}{e}{a}$
    \end{enumerate}
    then $de$ meets $ab$ in an event $f$ such that $\ord{a}{f}{b}$.
\end{theorem}
\begin{proof}
By the previous theorem (Theorem~2), the statement $\chain{a,f,b}$ of the Axiom of Collinearity (Axiom~\ref{ax:O6}) implies $\ord{a}{f}{b}$.
\end{proof}

The proof in Isabelle again follows Schutz closely. His proof, a single sentence quoting Axiom~\ref{ax:O6} and Theorem~2, is expanded upon merely by finding the precise paths to use in the Axiom of Collinearity (O6), namely $ac$ and $bc$.


\begin{lstlisting}
theorem (*3*) (in MinkowskiChain) collinearity:
  assumes tri_abc: "\<triangle> a b c"
      and path_de: "path de d e"
      and bcd: "[[b c d]]"
      and cea: "[[c e a]]"
    shows "(\<exists>f\<in>de\<inter>(path_of a b). [[a f b]])"
\end{lstlisting}

%% file: in/thms-3.3.tex
In the spirit of Theorem~3, Schutz continues to strengthen the statements made by his axioms. Theorem~4 (Boundedness of the Unreachable Set, see also Fig.~\ref{fig:thm4}) is concerned with restating the Axiom I7, which shares its name, in the context of the chain order established in Theorem~2. Schutz' proof is a one-liner referencing these two results.

\begin{theorem}[Boundedness of the Unreachable Set]\label{thm:4}
    Let Q be any path and let b be any event such that $b \notin Q$. Given events $Q_x \in Q \setminus Q(b, \emptyset)$ and $Q_y \in Q(b, \emptyset)$, there is an event $Q_z \in Q \setminus Q(b, \emptyset)$ such that
    \begin{enumerate}
        \item[(i)] $\ord{Q_x}{Q_y}{Q_z}$, and
        \item[(ii)] $Q_x \neq Q_z$.
    \end{enumerate}
\end{theorem}

\begin{figure}
    \centering
    \includegraphics[width=0.5\textwidth]{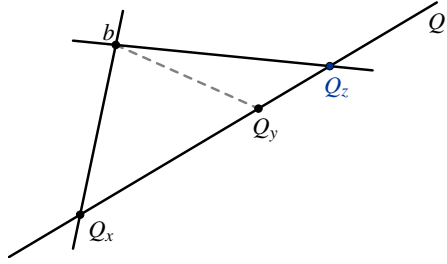}
    \caption{\label{fig:thm4}Boundedness of the Unreachable Set. Given $Q_x$ (reachable from $b$) and $Q_y$ (unreachable from $b$), Theorem~\ref{thm:4} obtains $Q_z$ (reachable from $b$). Axiom~\ref{ax:I7} furthermore states that all three events must be part of a finite chain.}
\end{figure}

Formalisation is again very simple, and in fact, Theorem~\ref{thm:4} can be proven in one step by Isabelle's \lstinline|metis|. The only results needed for this (apart from the theorem assumptions, I7, and Theorem~2) are the definition of chains and a corollary of Theorem~2 without explicit indices (\lstinline|fin_ch_betw|).

\begin{lstlisting}
lemma fin_ch_betw:
  assumes "[f[a..b..c]X]"
  shows "[[a b c]]"
\end{lstlisting}

\enlargethispage{5\baselineskip}
\begin{lstlisting}
theorem (*4*) (in MinkowskiUnreachable) unreachable_set_bounded:
  assumes path_Q: "Q \<in> \<P>"
      and b_nin_Q: "b \<notin> Q"
      and b_event: "b \<in> \<E>"
      and Qx_reachable: "Qx \<in> Q - \<emptyset> Q b"
      and Qy_unreachable: "Qy \<in> \<emptyset> Q b"
  shows "\<exists>Qz\<in>Q - \<emptyset> Q b. [[Qx Qy Qz]] \<and> Qx \<noteq> Qz"
  using assms I7 order_finite_chain fin_long_chain_def
  by (metis fin_ch_betw)
\end{lstlisting}

Theorem~5 allows one to generate additional events, given an event and a path: a second event on the same path, and a reachable event outside the path. After Theorem~3, this is the next more involved proof of the monograph.
The events provided by Theorem~\ref{thm:5} form a triangle of paths, thus enabling very geometric proofs of several lemmas leading up to Theorem~\ref{thm:9}.
These lemmas are, in practice, amongst the most important results for this work, both practically and conceptually, allowing to conclude new betweenness relations from existing ones (similarly to Axiom \ref{ax:O4}).

\begin{theorem}[First Existence Theorem]\label{thm:5}
    Given a path $Q$ and an event $a \in Q$, there is
    \begin{enumerate}
        \item[(i)] an event $b \in Q$ with $b$ distinct from $a$, and
        \item[(ii)] an event $c \notin Q$ and a path $ac$ (distinct from $Q$).
    \end{enumerate}
\end{theorem}

Schutz first shows that there is an event $d$ outside the path $Q$. This is done by contradiction, i.e. by showing that there cannot be a path containing all events (by Axiom I3, this would be the only existing path). We encapsulate this statement, the crux of the proof of Theorem~5(i), in a helper lemma.

\begin{lstlisting}
lemma (in MinkowskiUnreachable) only_one_path:
  assumes path_Q: "Q \<in> \<P>"
      and all_inQ: "\<forall>a\<in>\<E>. a \<in> Q"
      and path_R: "R \<in> \<P>"
  shows "R = Q"
\end{lstlisting}

In addition to Axiom I3, we require I5 in order to prove this, which Schutz misses out. The proof is again by contradiction: If a path $R$ exists that is not $Q$, then, since $Q$ is the set of all events, $R \subset Q$. A contradiction to I3, the Axiom of Uniqueness (of paths), can only be obtained if there are two events on $R$, which is guaranteed by Axiom I5.
The remainder of Theorem~5(i) follows Schutz, using I4 to contradict $Q=\E$, and I5 again to obtain the required event $b$.

The second statement of Theorem~5 is proved as in the original prose. In particular, now that we have two events, a second path is implied by Axiom I2, as in the statement below.

\begin{lstlisting}
lemma ex_crossing_path:
  assumes path_Q: "Q \<in> \<P>"
  shows "\<exists>R\<in>\<P>. R \<noteq> Q \<and> (\<exists>e. e \<in> R \<and> e \<in> Q)"
\end{lstlisting}

Then our proof follows the case split made by Schutz: either $e=a$ or not. The latter case becomes a little longer than in prose, but there are no surprises. Both cases use Axiom I5 to obtain the desired reachable event $c$. The final pair of statements for Theorem~5 is listed below.

\begin{lstlisting}
theorem (*5i*) ge2_events:
  assumes path_Q: "Q \<in> \<P>"
      and a_inQ: "a \<in> Q"
  shows "\<exists>b\<in>Q. b \<noteq> a"
\end{lstlisting}

\begin{lstlisting}
theorem (*5ii*) ex_crossing_at:
  assumes path_Q: "Q \<in> \<P>"
      and a_inQ: "a \<in> Q"
  shows "\<exists>ac\<in>\<P>. ac \<noteq> Q \<and> (\<exists>c. c \<notin> Q \<and> a \<in> ac \<and> c \<in> ac)"
\end{lstlisting}

%% file: in/thms-3.4.tex
Theorem~\ref{thm:6} goes a little further in justifying our intuition of paths as line-like objects by showing they are infinite. This also gives us the means to always find more events on a path.

\begin{theorem}[Prolongation]\label{thm:6}
    \begin{enumerate}
    \item[(i)] If $a,b$ are distinct events of a path $Q$, then there is an event $c \in Q$ such that $\ord{a}{b}{c}$.
    \item[(ii)] Each path contains an infinite set of distinct events.
    \end{enumerate}
\end{theorem}

Schutz' proof \cite[p.~21]{schutz1997} of the first part is straightforward, and remains so in Isabelle: the formal proof reads almost exactly like Schutz' prose. Theorem 5(ii) provides an event \linebreak$e \notin Q$ and a path $ae$. Axiom~I5 then guarantees existence of an event $f \in ae$ that is unreachable from $b$; thus $b \in Q(f, \emptyset)$. Theorem~4 delivers the desired event $c$.

\begin{lstlisting}
lemma (in MinkowskiSpacetime) prolong_betw2:
  assumes path_Q: "Q \<in> \<P>"
      and a_inQ: "a \<in> Q"
      and b_inQ: "b \<in> Q"
      and ab_neq: "a \<noteq> b"
  shows "\<exists>c\<in>Q. [[a b c]]"
\end{lstlisting}

While the second part of Theorem~\ref{thm:6} can be proven almost by inspection by the reader, it is much trickier to formalise.
Schutz says that ``By the preceding theorem [...] part (i), Theorem~1, and induction, the path $Q$ contains an infinite set of distinct events''. Our problem is to formalise this list of results into an inductive proof that can be checked by Isabelle.
This involves thinking about how to translate from induction on a natural number to infinity, what exactly the induction variable should be,
and properly applying Isabelle's induction rule.

The main proof is by induction on the cardinality of a subset $X \subseteq Q$, and is encapsulated by the helper lemma \lstinline|finite_path_has_ends|, which allows us to choose two elements $a,b$ of a set of events on a path $Q$, such that all other elements of that set are between $a$ and $b$.
\jeptodo[fancyline]{Shall we say anything about alternative approaches we tried or thought about here? So far we haven't discussed any dead-ends, just end-products.}



\begin{lstlisting}
lemma finite_path_has_ends:
  assumes "Q \<in> \<P>"
      and "X \<subseteq> Q"
      and "finite X"
      and "card X \<ge> 3"
    shows "\<exists>a\<in>X. \<exists>b\<in>X. a \<noteq> b \<and> (\<forall>c\<in>X. a \<noteq> c \<and> b \<noteq> c \<longrightarrow> [[a c b]])"
\end{lstlisting}

These events will later be used to apply the first part of Theorem~\ref{thm:6}. A sample listing of the proof is given below. We begin by applying the induction hypothesis to identify the edges $a$ and $b$ of the set $Y$.

\begin{lstlisting}
proof (induct "card X - 3" arbitrary: X)
(*...*)
case IH: (Suc n)
  obtain Y x where X_eq: "X = insert x Y" and "x \<notin> Y"
    by (meson IH.prems(4) Set.set_insert three_in_set3)
  (*...*)
  obtain a b where ab_Y: "a \<in> Y" "b \<in> Y" "a \<noteq> b"
             and Y_ends: "\<forall>c\<in>Y. (a \<noteq> c \<and> b \<noteq> c) \<longrightarrow> [[a c b]]"
    using IH(1) [of Y] IH.prems(1-3) X_eq by auto
\end{lstlisting}

\noindent
The rest of the proof treats each possible ordering of the additional event $x \in X \setminus Y$ with $a$ and $b$, to identify the extremal events of the larger set $X$.

\begin{lstlisting}
  consider "[[a x b]]" | "[[x b a]]" | "[[b a x]]" <proof>
  thus ?case
  proof (cases)
    (*...*)
    assume "[[x b a]]"
    { fix c
      assume "c \<in> X" "x \<noteq> c" "a \<noteq> c"
      then have "[[x c a]]" <proof>
    }
    thus ?thesis
      using X_eq \<open>[[x b a]]\<close> ab_Y(1) abc_abc_neq insert_iff
      by force
  qed
qed
\end{lstlisting}

We can now prove that the cardinality of a path cannot be finite. The cases for less than three events on a path are dispensed with separately, using Theorem~5
(as is hinted in the prose we gave above).
For any set of events on a path $Q$, we can use the lemma \lstinline|finite_path_has_ends| to obtain events $a,b$ such that all other elements of that set are between $a$ and $b$. Assuming any finite, non-zero cardinality%
\footnote{The default way of treating cardinality in Isabelle is using natural numbers only. In this formalism, infinite sets are given cardinality $0$.}
of the path $Q$, the prolongation obtained from these two endpoint events using Theorem~6(i) can be used to obtain a contradiction. Thus we conclude that the cardinality of any path must be $0$.

\begin{lstlisting}
lemma path_card_nil:
  assumes "Q\<in>\<P>"
  shows "card Q = 0"
\end{lstlisting}

However, we also know that the empty set is not a path (see Sec.~\ref{sec:bgr:isaind}), thus all paths must be infinite. The formalised result is slightly more simply stated than Schutz' \textit{``Each path contains an infinite set of distinct events''}, since any path that contains an infinite subset must be infinite (and conversely, since paths contain only events, an infinite path must have infinite subsets of events). We are not sure why Schutz did not settle for this seemingly more elegant formulation.

\begin{lstlisting}
theorem (*6ii*) infinite_paths:
  assumes "P\<in>\<P>"
  shows "infinite P"
\end{lstlisting}

%% file: in/thms-3.5.tex
The Second Collinearity Theorem extends the First (Theorem~\ref{thm:3}) by adding the ordering $\ord{d}{e}{f}$ to the conclusion.

\begin{theorem}[Second Collinearity Theorem]\label{thm:7}
    In the notation of collinearity (Axiom O6),
    $$\ord{a}{f}{b} \text{ and } \ord{d}{e}{f} \;.$$
    That is, given a kinematic triangle $\triangle abc$ with $\ord{b}{c}{d}$ and $\ord{c}{e}{a}$, if there is a path $de$ then on the path $de$ there is an event $f$ such that
    $$\ord{a}{f}{b} \text{ and } \ord{d}{e}{f} \;.$$
\end{theorem}

The proof begins where Theorem~\ref{thm:3} left off (see also Fig.~\ref{fig:O6}), i.e. we need prove only $\ord{d}{e}{f}$. Since $f \in de$ (so $d\neq e$),
there are six possible relations between the events $d,e,f$, summarised in Isabelle by \hbox{\lstinline|some_betw2|.}

\begin{lstlisting}
lemma (in MinkowskiBetweenness) some_betw2:
  assumes path_Q: "Q \<in> \<P>"
      and a_inQ: "a \<in> Q" and b_inQ: "b \<in> Q" and c_inQ: "c \<in> Q"
  shows "a = b \<or> a = c \<or> b = c \<or> [[a b c]] \<or> [[b c a]] \<or> [[c a b]]"
\end{lstlisting}

Since $de$ defines a path, we know $d\neq e$. Either one of the remaining equalities would imply that $a,b,c$ are on the same path. For example, if $e=f$, then $\ord{c}{e}{a}$ and $\ord{a}{f}{b}$ (from Theorem~\ref{thm:3}) imply $a,b,c$ are on the path $af$. Since paths are unique, and the definition of kinetic triangles includes distinctness of the three defining paths, this contradicts $\triangle abc$.

\begin{lstlisting}
lemma triangle_diff_paths:
  assumes tri_abc: "\<triangle> a b c"
  shows "\<not> (\<exists>Q\<in>\<P>. a \<in> Q \<and> b \<in> Q \<and> c \<in> Q)"
\end{lstlisting}

There are two remaining possibilities to falsify. Both sub-proofs by contradiction follow the same layout, so we present only the case $\ord{e}{f}{d}$. We first show that $\triangle dce$. Schutz takes this as fact, but Isabelle requires us to demonstrate $de \neq ce$. Since $\ord{b}{c}{d}$ and $\ord{c}{e}{a}$, and $bc \neq ac$, we know that $\lnot \ord{c}{d}{e}$, which establishes the kinematic triangle.
Then $\ord{c}{e}{a}$ and $\ord{e}{f}{d}$ satisfy condition (ii) of Theorem~\ref{thm:3} and we obtain an event $x \in af$ with $\ord{d}{x}{c}$. Uniqueness (I3) then gives $x=b$, hence $\ord{d}{b}{c}$, which contradicts the assumption $\ord{b}{c}{d}$ via Theorem~\ref{thm:1}.
Discounting the possibility $\ord{f}{d}{e}$ in an analogous manner, we are left with only $\ord{d}{e}{f}$.

The formalisation follows Schutz rather easily, with only the proof $\triangle dce$ requiring an extra step. Notice that this theorem could have been proved at the same time as Theorem~\ref{thm:3},
but encapsulating the First Collinearity Theorem allows us to use it multiple times throughout the proof of the Second Collinearity Theorem.

\begin{lstlisting}
theorem (*7*) (in MinkowskiChain) collinearity2:
  assumes tri_abc: "\<triangle> a b c"
      and bcd: "[[b c d]]"
      and cea: "[[c e a]]"
      and path_de: "path de d e"
  shows "\<exists>f\<in>de. [[a f b]] \<and> [[d e f]]"
\end{lstlisting}

%% file: in/thms-3.6a.tex
This section gives the chapter its name, and will allow us to work much more freely with the betweenness relation, bringing it closer to the intuition we have from Euclidean geometry. Theorem~\ref{thm:8} is a preliminary result, but provides an intuitive piece of information about kinematic triangles. Theorem~\ref{thm:9} and Theorem~\ref{thm:10} establish finite subsets of paths as totally ordered sets\footnote{
We have in fact proven that not just finite subsets, but paths themselves are totally ordered. This proof uses a binary order derived from betweenness and uses definitions of HOL-Algebra; since this dependency is not required anywhere else, and Schutz does not introduce binary order until Theorem 29, we refrain from giving this result here.}.
The proof of Theorem~\ref{thm:9} hinges on three lemmas that are, to any practical purpose, as important as any result of this chapter, and allow us to work with orderings of overlapping sets of events.

Theorem~\ref{thm:8} presupposes the easy result (not explicitly mentioned by Schutz) that $\triangle abc$ implies no betweenness ordering of $a,b,c$ exists, and extends it to events on the paths defining the triangle (rather than its vertices)%
\footnote{The equivalence between there being some ordering of $a,b,c$, and all three events being on a path is established by Axiom~\ref{ax:O1} and Axiom~\ref{ax:O5}.}.
Using some geometric intuition, Theorem~\ref{thm:8} might be likened to the statement
that no path can cross all three sides of a kinematic triangle internally.

\begin{theorem}\label{thm:8}
    Given a kinematic triangle $\triangle abc$ with events $a', b', c'$ such that $\ord{a}{b'}{c}$, $\ord{b}{c'}{a}$, and $\ord{c}{a'}{b}$, then there is no path which contains $a'$, $b'$ and $c'$.
\end{theorem}

Schutz first notes that $a',b',c'$ are distinct from $a,b,c$ and from each other, by the orderings assumed in Theorem~\ref{thm:8} and Axioms \ref{ax:O3} and \ref{ax:I3}. Notice also that existence of a path containing three events would imply some ordering of these events.

The proof of Theorem~\ref{thm:8} is then by contradiction. We assume there is some ordering of $a',b',c'$, and we examine the cases one-by-one, starting with $\ord{a'}{b'}{c'}$. A small sub-proof by contradiction shows there is no path $Q \ni a',b,c'$, since Axiom~\ref{ax:I3} would imply $a, c \in Q$, placing the three vertices of the triangle on the same path.
Thus $a', b, c'$ form a kinematic triangle.
Much like we used Theorem~\ref{thm:3} in the last proof, we now apply Theorem~\ref{thm:7} to $\triangle a'bc'$, and obtain $x = ab' \cap a'b$ where $\ord{a'}{x}{b}$. Obtaining $c = ab' \cap a'b$ is slightly longer in Isabelle than in Schutz,
but follows from the assumed orderings $\ord{a}{b'}{c}$ and $\ord{c}{a'}{b}$. The contradiction is between $\ord{a'}{x}{b} \implies \ord{a'}{c}{b}$ and $\ord{c}{a'}{b}$ (cf Theorem~\ref{thm:1}).

While the proof has closely followed the prose so far, Schutz now goes on to state simply that ``cyclic interchange of the symbols $a$, $b$, $c$ (and $a'$, $b'$, $c'$) throughout the proof'' \cite[p.~23]{schutz1997} proves the remaining cases. In Isabelle, this interchange is done explicitly, by reproducing the same proof with different event orderings. However, this is the first time that we encountered ordering symmetry of this sort, where one has to consider multiple equivalent cases depending not on essential qualities of events, but their names, and the inconsequential (or arbitrary) ordering that results from this naming scheme. This kind of reasoning is often employed in mathematics, and might be announced simply as ``without loss of generality, let $\ord{a'}{b'}{c'}$''. While in this case, the complete proof is still less than 150 lines of proof script, this kind of redundancy becomes disproportionate later on, and we refer the reader to our treatment of Theorem~\ref{thm:14} for details. The mechanised Theorem~\ref{thm:8} is given below.

\begin{lstlisting}
theorem (*8*) (in MinkowskiChain) tri_betw_no_path:
  assumes tri_abc: "\<triangle> a b c"
      and ab'c: "[[a b' c]]"
      and bc'a: "[[b c' a]]"
      and ca'b: "[[c a' b]]"
  shows "\<not> (\<exists>Q\<in>\<P>. a' \<in> Q \<and> b' \<in> Q \<and> c' \<in> Q)"
\end{lstlisting}

Theorem~\ref{thm:9} is the base case for the inductive Theorem~\ref{thm:10}. One might compare these two results to parts (i) and (ii) of Theorem~\ref{thm:6}, but the induction is more complicated in the case of Theorem~\ref{thm:10}, and hides a few more surprises when attempting a formalisation.

\begin{theorem}\label{thm:9}
    Any four distinct events on a path form a chain, so they may be represented by the symbols $a$, $b$, $c$, $d$ in such a way that $\chain[\;]{a,b,c,d}$.
\end{theorem}

This result extends the Axiom~\ref{ax:O5}, with a chain being the appropriate generalisation of betweenness via Theorem~\ref{thm:2}. Thus the main point of Theorem~\ref{thm:9} is to do with overlapping betweenness relations between subsets of three out of four events. The proof is split into three lemmas that, together, allow us to propagate betweenness relations along a chain. The first one is the hardest to prove: the other two (and several similar results not printed in Schutz) follow from it easily.

\begin{lemma}\label{lem:1}
    If $\ord{a}{b}{c}$ and $\ord{a}{b}{d}$ and $c \neq d$ then either $\ord{b}{c}{d}$ or $\ord{b}{d}{c}$.
\end{lemma}


\begin{lstlisting}
lemma abc_abd_bcdbdc:
  assumes abc: "[[a b c]]"
      and abd: "[[a b d]]"
      and c_neq_d: "c \<noteq> d"
  shows "[[b c d]] \<or> [[b d c]]"
\end{lstlisting}

To prove \lstinline|abc_abd_bcdbdc|, we follow Schutz fairly closely, with the top layer being a proof by contradiction together with
$\neg [d b c] \rightarrow [b c d] \lor [b d c]$, which is obtained by noting that path uniqueness (Axiom~\ref{ax:I3}) and \lstinline|abc_ex_path| (Axiom~\ref{ax:O1}) imply that $b, c, d$ all lie on the same path, and thus must be in some betweenness relationship (Axiom~\ref{ax:O5}).
We thus assume $[d b c]$ and derive a contradiction by constructing several kinematic triangles, whose interaction with each other leads to a contradiction with Theorem~\ref{thm:8} (\lstinline|tri_betw_no_path|).


We obtain the basic geometric ingredients: first a path containing $a$ and $b$. 
Given a path $ab$ and an event $a$ on it, Theorem~\ref{thm:5}
provides a different path $S$.
Using the existence of unreachable events (Axiom~\ref{ax:I5}) and the boundedness of the unreachable set (Theorem~\ref{thm:4}), we obtain $e \in S \setminus \left\lbrace a \right\rbrace$ (so we can rename $S=ae$) and a path $be$.


The difficulty of translating Schutz' approach to the remaining proof into Isabelle, is in his conditional assignment of events to the variables he calls $c*$, $d*$, and $f*$.
For example, Schutz defines $d'$ as ``If there is a path $de$ we let $d'=d$. Otherwise [Theorem~\ref{thm:4}] implies the existence of an event $d'$ such that [...]''. This would require us to consider both cases for each statement involving $d'$ in the remainder of the proof; we found this to be cumbersome in Isabelle.
\footnote{These become {\lstset{basicstyle=\ttfamily\footnotesize}\lstinline|c', d', f'|} in our formal proof since the $*$-affix is reserved in Isabelle.}
We abstract this difficulty into lemmas called \lstinline|exist_c'd'| and \lstinline|exist_f'|. Several case splits need to be considered, but have no further importance outside of these lemmas: thus we separate them from the main proof. Notice that \lstinline|exist_c'd'| and \lstinline|exist_f'| are trivial in a highly non-obvious fashion: since they are to be used inside a proof by contradiction, their assumptions already imply \lstinline|False|, which implies anything. This implication, however, is complex enough not to be detected by Isabelle's automatic tools, nor was it by us upon inspection. The assumptions on both lemmas are equivalent to the obtained facts in the main proof at the point of their use. 


\begin{lstlisting}
lemma exist_c'd':
  assumes abc: "[[a b c]]"
      and abd: "[[a b d]]"
      and dbc: "[[d b c]]"
      and path_S: "path S a e"
      and path_be: "path be b e"
      and S_neq_ab: "S \<noteq> path_of a b"
    shows "\<exists>c' d'. [[a b d']] \<and> [[c' b a]] \<and> [[c' b d']] \<and>
                   path_ex d' e \<and> path_ex c' e"
\end{lstlisting}

Schutz' proof
considers nested case splits ``in parallel'', jumping between cases for each statement in the flow of the main proof. We instead just abstract proofs of existential propositions with all the properties we need into the lemmas \lstinline|exist_c'd'| and \lstinline|exist_f'|, and require no case splits in the main proof. We find this setup both easier to formalise in Isar, and easier to understand for the reader. In this case, practical concerns towards a neater formalisation lead, we believe, to a less convoluted, more modular proof.

A structural outline is provided for the proof body of \lstinline|exist_c'd'|, but most of the individual steps are omitted. Notice the case splits according to whether paths between certain events exist, which reproduce those of Schutz.
\begin{lstlisting}
proof (cases "path_ex d e")
  let ?ab = "path_of a b"
  have path_ab: "path ?ab a b" <proof>
  { case True
    then obtain de where "path de d e" by blast
    (*...*)
    thus ?thesis
    proof (cases "path_ex c e")
      case True (*...*)
    next
      case False
      obtain c' c'e where "c'\<in>?ab \<and> path c'e c' e \<and> [[b c c']]"
        using unreachable_bounded_path <proof>
      (*...*)
    qed
  } {
    case False
    obtain d' d'e where d'_in_ab: "d' \<in> ?ab"
                    and bdd': "[[b d d']]" and "path d'e d' e"
      using unreachable_bounded_path <proof>
    thus ?thesis
    proof (cases "path_ex c e")
      (*...*)
  }
qed
\end{lstlisting}

Using the lemma \lstinline|unreachable_bounded_path| above, we replace Schutz' more vague statement of ``the Boundedness of the Unreachable Set (Th.4) implies''.
While this lemma relies on Theorem~\ref{thm:4} and the assumptions of \lstinline|exist_c'd'| only (excluding definitions), several steps are needed in Isabelle to derive this result. 
The lemma \lstinline|exist_f'|, which is proved similarly, is omitted here. 


From here on, the proof follows Schutz, who in turn follows Veblen \cite[p.357]{veblen1904}.
The idea is to find three events on the path \lstinline|f'b|, obtained from \lstinline|exist_f'|, that lie on different sides of the kinematic triangle $\triangle ead'$. This gives a contradiction to Th.8: no path can cross all three sides of a kinematic triangle. These events, $g$ and $h$, as well their ordering relations with $f',b,e$, are obtained by applying Theorem~\ref{thm:7} to two different kinematic triangles, outlined in Figure~\ref{fig:lem1}.
Now, $\ord{a}{h}{e}$, $\ord{d'}{g}{e}$, $\ord{a}{b}{d'}$ together imply that $b,g,h$ lie on different segments of $\triangle ead'$. However, all three must lie on a path by $\ord{f'}{b}{h}$ and $\ord{f'}{b}{g}$, contradicting Theorem~\ref{thm:8}.
Thus we conclude Lemma 1.

\begin{figure}
    \centering
     \begin{subfigure}[b]{0.48\textwidth}
         \centering
         \includegraphics[width=\textwidth]{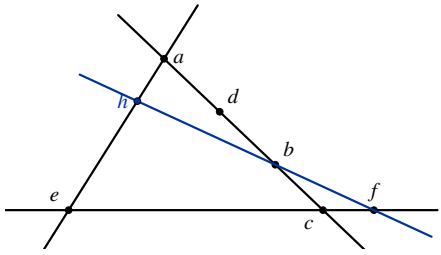}
         \caption{Triangle $\triangle ace$, ignoring the assumption $\ord{a}{b}{d}$ in the figure. Applying Theorem~\ref{thm:7} will yield the event $h$ with $\ord{a}{h}{e}$ and $\ord{f}{b}{h}$.}
         \label{fig:lem1_t-ace}
     \end{subfigure}
     \hfill
     \begin{subfigure}[b]{0.48\textwidth}
         \centering
         \includegraphics[width=\textwidth]{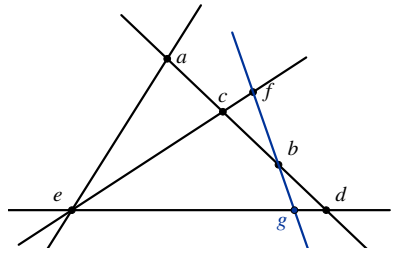}
         \caption{Triangle $\triangle dec$, ignoring the assumption $\ord{a}{b}{c}$ in the figure. Applying Theorem~\ref{thm:7} will yield the event $g$ with $\ord{d}{g}{e}$ and $\ord{f}{b}{g}$.}
         \label{fig:lem1_t-dec}
     \end{subfigure}
    \caption{Visualisation of an intermediate state in the proof of Lemma 1 \cite[pp.~23-24]{schutz1997}. For simplicity, all primed (or starred, in Schutz' prose) variables are equal to their unprimed counterparts. Since the proof is by contradiction, and plane geometry obeys the axioms of order and incidence, it is impossible to draw a correct figure: thus each of these constructions ignores one assumption in order to visualise different triangles to which Theorem~\ref{thm:7} is applied.}
    \label{fig:lem1}
\end{figure}

\begin{lemma}\label{lem:2}
    If $\ord{a}{b}{c}$ and $\ord{a}{b}{d}$ and $c \neq d$ then either $\ord{a}{c}{d}$ or $\ord{a}{d}{c}$.
\end{lemma}

\begin{lemma}\label{lem:3}
    If $\ord{a}{b}{c}$ and $\ord{a}{c}{d}$ then $\ord{b}{c}{d}$.
\end{lemma}

The two remaining lemmas follow quickly from Theorem~\ref{thm:1}, Axiom~\ref{ax:O4}, and Lemma~\ref{lem:1}. In addition, we prove several further, similar results that follow readily too, such as
\[\ord{a}{b}{c} \wedge \ord{b}{c}{d} \implies \ord{a}{b}{c}\;, \text{ and}\]
\[\ord{a}{b}{e} \wedge \ord{a}{d}{e} \wedge \ord{b}{c}{d} \implies \ord{a}{c}{e}\;.\]

Theorem~\ref{thm:9} is now rather easy to prove, and we are able to follow Schutz' prose closely. However, we state the result in a different way: the prosaic ``\dots may be represented by the symbols $a$, $b$, $c$, $d$ in such a way that \dots'' from Theorem~\ref{thm:9} above is more easily expressed in Isabelle as a property of the set of all four events.

\begin{lstlisting}
theorem (*9*) chain4:
  assumes path_Q: "Q \<in> \<P>"
      and inQ: "a \<in> Q" "b \<in> Q" "c \<in> Q" "d \<in> Q"
      and abcd_neq: "a \<noteq> b \<and> a \<noteq> c \<and> a \<noteq> d \<and> b \<noteq> c \<and> b \<noteq> d \<and> c \<noteq> d"
    shows "ch {a,b,c,d}"
\end{lstlisting}

Even though our statement is different, we are able to follow Schutz' proof in the main aspects, and briefly explain it here.
A start of the chain, say $\ord{a}{b}{c}$, is easily obtained from Axiom~\ref{ax:O5}. The remaining element $d$ must then obey, also by \ref{ax:O5}, one of the three orderings $\ord{d}{a}{b}$, $\ord{a}{d}{b}$, or $\ord{a}{b}{d}$. In each case, the Lemmas \ref{lem:1} and \ref{lem:3} provide the remaining ordering we need for the chain of four elements.


In the actual formalisation, since events are named $a,b,c,d$ in the assumptions, we may not simply assume $\ord{a}{b}{c}$ and expect Isabelle to notice this happens, as mathematicians might say, without loss of generality. Instead, we obtain new events $a',b',c' \in \left\lbrace a,b,c,d \right\rbrace$ such that $\ord{a'}{b'}{c'}$. The rest of the proof follows as above, and concludes Theorem~\ref{thm:9}.

%% file: in/thms-3.6b.tex

\begin{theorem}\label{thm:10} 
Any finite set of distinct events of a path forms a chain. That is, any set of $n$ distinct events can be represented by the notation $a_1,a_2,\dots,a_n$ such that\vspace{-7pt}
\[[a_1 \; a_2 \dots a_n]\;.\]
\end{theorem}

There is, of course, nothing special about a set of four elements on a path: one would expect the result of Theorem~\ref{thm:9} to extend to larger sets too. Theorem~\ref{thm:10} proves that this is so.
Mechanising it was a major undertaking.
One problem was due to the definition of \lstinline|ordering| we were using initially \cite{palmer2017}: the chain definition used in most of our early results is stronger than Schutz'.
As shown in Sec.~\ref{sec:order-fin-chain}, this leads to a free proof of Theorem~\ref{thm:2}. But such things always come with a price:
Schutz' proof of Theorem~\ref{thm:10} only aims at a local chain. If we want to be consistent with our previous proofs in Isabelle (which use total chains, \lstinline|long_ch_by_ord|), we need this local chain to become a total chain, which essentially means going through all the steps of Schutz' proof for Theorem~\ref{thm:2}. This is why we defined a new local \lstinline|ordering2|, and proved \lstinline|order_finite_chain2| in Sec.~\ref{sec:order-fin-chain}.

Like for Theorem~\ref{thm:9}, we ignore the second sentence of Schutz' formulation, which essentially restates the first, but is harder to express in Isabelle.
Our statement differs from Schutz in another way. 
We make explicit the condition that any chain needs to have at least two elements (by definition): thus it isn't \textit{every} finite set of events that qualifies. This is left implicit in Schutz' notation, where $[a_1 \; a_2 \dots a_n]$ implies $n \geq 2$, leaving his first sentence imprecise. This condition on the cardinality of $X$ implies finiteness in Isabelle.

\begin{lstlisting}
theorem (*10*) path_finsubset_chain:
  assumes "Q \<in> \<P>"
      and "X \<subseteq> Q"
      and "card X \<ge> 2"
  shows "ch X"
\end{lstlisting}


The proof is by induction, as in Schutz \cite{schutz1997}. 
Notice Schutz uses a four-element chain as the base case, 
so we have to provide two (simple) extra cases: two- and three-element sets.
A two-element chain is just a set of two points on a path, thus a two-event set $X$ satisfies the definition of chains immediately. A set $X$ with three events $a,b,c$, all of them on a path, must be a chain because $a,b,c$ are in some betweenness relation by axiom O5. Both of these are omitted from the listing, and we move on to the induction.




The base case of $|X|=4$ follows directly from Theorem~9: it states that a set of four events on a path forms a chain.
Schutz' induction proceeds by assuming a chain of $n$ events, and adds an extra event. We instead fix the number of events in the set to the successor of the induction variable $n$ (where $n=|X|-4$ because Isabelle induction starts at $n=0$, see Sec.~\ref{sec:bgr:isaind}). Thus we obtain a new set by removing an element, and argue this new set must be a chain by the induction hypothesis \lstinline|IH|. We remove some overall indentation for legibility.


\begin{lstlisting}
case IH: (Suc n)
  then obtain Y b where X_eq: "X = insert b Y" and "b \<notin> Y"
    by (metis Diff_iff card_eq_0_iff finite.cases insertI1 insert_Diff_single not_numeral_le_zero)
  have "card Y \<ge> 4" "n = card Y - 4"
    using IH.hyps(2) IH.prems(4) X_eq \<open>b \<notin> Y\<close> by auto
  then have "ch Y"
    using IH(1) [of Y] IH.prems(3,4) X_eq assms(1) by auto
  then obtain f where f_ords: "long_ch_by_ord f Y"
    using ch_long_if_card_ge3 \<open>4 \<le> card Y\<close> by fastforce
\end{lstlisting}

This places us in the setting of Schutz' proof: we have a chain $Y$, indexed by $f$, of at least four events, and a set $X$ containing one extra event $b$ in addition to this chain. We now introduce variable names that agree with those of Schutz, e.g.\ $a_1 = f(0)$. In terms of our indexing function, the subscripts of those variables are shifted, but it allows us to reproduce his prose (given further below) more faithfully.

\begin{lstlisting}
obtain a\<^sub>1 a a\<^sub>n where long_ch_Y: "[f[a\<^sub>1..a..a\<^sub>n]Y]"
  using get_fin_long_ch_bounds Y_def f_def fin_X
  by fastforce
hence bound_indices: "f 0 = a\<^sub>1 \<and> f (card Y - 1) = a\<^sub>n"
  by (simp add: fin_long_chain_def)
\end{lstlisting}

The remaining proof is structured into the same three cases Schutz considers. We obtain the three possible betweenness relations that the three events above can be in, and consider each in turn.

\pagebreak
\begin{lstlisting}
consider "[[b a\<^sub>1 a\<^sub>n]]" | "[[a\<^sub>1 a\<^sub>n b]]" | "[[a\<^sub>n b a\<^sub>1]]" <proof>
thus "ch X"
proof (cases)
  (* case (i) *)
  assume "[[b a\<^sub>1 a\<^sub>n]]"
  obtain g where "g=(\<lambda>j::nat. if j\<ge>1 then f (j-1) else b)"
    by simp
  hence "[g[b..a\<^sub>1..a\<^sub>n]X]"
    using chain_append_at_left_edge (*...*) by blast
  thus "ch X"
    unfolding ch_def ch_by_ord_def using fin_long_chain_def by auto
\end{lstlisting}


The main proof steps needed for this first case, with $\ord{b}{a_1}{a_n}$, are inside \linebreak\hbox{\lstinline|chain_append_at_left_edge|}. Schutz' prose for this case is given below \cite[p.~25]{schutz1997}.

\jtodo[fancyline]{indent using quote env?}
\begin{proof}[Case (i)]
By the inductive hypothesis and Theorem~\ref{thm:2} we have $\ord{a_1}{a_2}{a_n}$, so the previous theorem [Theorem~\ref{thm:9}] implies that $[b \; a_1 \; a_2 \; a_n]$ which implies that $\ord{b}{a_1}{a_2}$. Thus $b$ is an element of a chain $[a^{*}_{1} \; a^{*}_{2} \; \dots \; a^{*}_{n+1}]$ where $a^{*}_{1} = b$ and (for $j\in \left\lbrace 2,\dots,n+1\right\rbrace$) $a^{*}_{j} := a_{j-1}$.
\qed
\end{proof}

We skip the step involving $[b \; a_1 \; a_2 \; a_n]$, employing instead an alternative ordering relation
\lstinline|abd_bcd_abc|, proving $$\ord{a}{b}{d} \text{ and } \ord{b}{c}{d} \implies \ord{a}{b}{c}\;.$$
This result is not given in Schutz, but it follows readily from the lemmas provided in the proof of Theorem~\ref{thm:9}.
We could have formulated a four-element chain with an explicit indexing function to follow Schutz' more closely, but since that requires multiple extra definitions, we decided this way was easier to read in Isabelle. We give a heavily cut listing of the proof below (remember that $a_2$ becomes $f(1)$).

\begin{lstlisting}
lemma chain_append_at_left_edge:
  assumes long_ch_Y: "[f[a\<^sub>1..a..a\<^sub>n]Y]"
      and bY: "[[b a\<^sub>1 a\<^sub>n]]"
    fixes g defines g_def: "g \<equiv> (\<lambda>j. if j\<ge>1 then f (j-1) else b)"
    shows "[g[b .. a\<^sub>1 .. a\<^sub>n](insert b Y)]"
proof -
  (*...*)
  hence "[[a\<^sub>1 (f 1) a\<^sub>n]]"
    using order_finite_chain fin_long_chain_def long_ch_Y
    by auto
  hence "[[b a\<^sub>1 (f 1)]]"
    using bY abd_bcd_abc by blast
\end{lstlisting}

Schutz' final sentence implies an indexing function that is equal to our
$g$, and his statement requires manual proofs of multiple chain properties regarding indexing and betweenness in Isabelle (namely, those that make up the definition of \lstinline|ordering2|, cf Sec.~\ref{sec:order-fin-chain}). Notice that this is where Theorem~\ref{thm:2} comes in for us, in the guise of \lstinline|ch_equiv| (see Sec.~\ref{sec:order-fin-chain}): Schutz only shows that a single betweenness relation holds between $b$ and adjacent elements. It is Theorem~\ref{thm:2} that allows us to extend this to betweenness relations involving any events on the (finite) chain, and obtain a total chain, thus proving our goal.
    
\begin{lstlisting}
  have "ordering2 g betw X" <proof>
  hence "long_ch_by_ord2 g X"
    using points_in_chain (*...*) by blast
  hence "long_ch_by_ord g X"
    using ch_equiv fin_X by blast
\end{lstlisting}

We now go back to Theorem~\ref{thm:10}'s induction. Two cases remain: $b$ being the middle element (ii), and $b$ being on the right (iii). Case (iii) is symmetric with case (i), and Schutz doesn't give an explicit proof of it. Instead of copy-pasting the entire proof for \linebreak\hbox{\lstinline|chain_append_at_left_edge|}, we therefore choose to use a different result, \lstinline|chain_sym|, to give a more interesting, shorter proof using symmetry. 

\begin{lstlisting}
lemma chain_sym:
  assumes "[f[a..b..c]X]"
    shows "[\<lambda>n. f (card X - 1 - n)[c..b..a]X]"
\end{lstlisting}

This relationship between a finite chain and its reversal is not explicitly mentioned in Schutz, an omission which leads to some complication also in our proof of Theorem~\ref{thm:13} (Sec.~\ref{sec:unreach-connected}). The lemma \lstinline|chain_sym| allows for a proof of Case~(iii) that makes use of Case~(i).

\begin{lstlisting}
lemma chain_append_at_right_edge:
  assumes long_ch_Y: "[f[a\<^sub>1..a..a\<^sub>n]Y]"
      and Yb: "[[a\<^sub>1 a\<^sub>n b]]"
    fixes g defines g_def: "g \<equiv> (\<lambda>j. if j \<le> (card Y - 1) then f j
    				     else b)"
    shows "[g[a\<^sub>1 .. a\<^sub>n .. b](insert b Y)]"
proof -
  (*...*)
  obtain f2 where f2_def: "[f2[a\<^sub>n..a..a\<^sub>1]Y]"
      "f2 = (\<lambda>n. f (card Y - 1 - n))"
    using chain_sym long_ch_Y by blast
  obtain g2 where g2_def: "g2 = (\<lambda>j. if j\<ge>1 then f2 (j-1) else b)"
    by simp
  have "[[b a\<^sub>n a\<^sub>1]]"
    using abc_sym Yb by blast
\end{lstlisting}

The functions $f_2$ and $g_2$ can be thought of as reversed versions of $f$ and $g$: if $f$ indexes a chain ``left-to-right'', $f_2$ counts ``right-to-left''. We can show $g_2$ orders $X$ into a chain using \lstinline|chain_append_at_left_edge|, and then reverse it again using \lstinline|chain_sym| to get $g_1$, which thus orders $X$. Finally, we show $g_1=g$, here in ordinary mathematical notation:


\begin{align*}
g_1(n) =
g_2(|X|-1-n) &=
\begin{cases}
	f_2(|X|-2-n) \;\; &\text{ if} \; |X|-1-n\ge1\\
	b \; &\text{ otherwise}
\end{cases}\\&=
\begin{cases}
	f(|Y|+1-|X|+n) \;\; &\text{ if} \; |X|-2\ge n\\
	b \; &\text{ otherwise}
\end{cases} \\&=
g(n)
\end{align*}


%
%

This concludes the cases of appending events at the end of a chain. Schutz' prose proof for the case of adding an event inside a chain is longer, and given below.

\begin{proof}[Case (ii)]
Let $k$ be the smallest integer such that $\ord{a_1}{b}{a_k}$.
Then the previous theorem [Theorem~\ref{thm:9}] implies either that $[a_1 \; a_{k-1} \; b \; a_k]$, or that $k=2$ so that $\ord{a_{k-1}}{b}{a_k}$.
If $k-2 \geq 1$ we have $\ord{a_{k-2}}{a_{k-1}}{a_k}$ which with $\ord{a_{k-1}}{b}{a_{k+1}}$ implies $[a_{k-2} \; a_{k-1} \; b \; a_k]$ by the previous theorem, while if $k+1 \leq n$ we have $\ord{a_{k-1}}{a_k}{a_{k+1}}$ which with $\ord{a_{k-1}}{b}{a_k}$ implies $[a_{k-1} \; b \; a_{k} \; a_{k+1}]$;
that is we have now shown that $\ord{a_{k-2}}{a_{k-1}}{b}$ (if $k-2 \geq 1$) and $\ord{a_{k-1}}{b}{a_k}$ and $\ord{b}{a_k}{a_{k+1}}$ (if $k+1 \leq n$) so that $b$ is an element of a chain $\chain{a_1^{*}\;a_2{*},a_{n+1}^{*}}$ where
\[ a_j^{*} = \begin{cases}
    a_j,        & j \leq k-1 \\
    b,          & j=k \\
    a_{j-1},    & j>k \;.
\end{cases} \]
\qed\end{proof}

Schutz' seemingly harmless first sentence ``Let $k$ be [\dots]'' requires a nontrivial existence proof in Isabelle.

\begin{lstlisting}
lemma (*for 10*) smallest_k_ex:
  assumes long_ch_Y: "[f[a\<^sub>1..a..a\<^sub>n]Y]"
      and b_def: "b\<notin>Y"
      and Yb: "[[a\<^sub>1 b a\<^sub>n]]"
    shows "\<exists>k>0. [[a\<^sub>1 b (f k)]] \<and> k < card Y \<and>
                 \<not>(\<exists>k'<k. [[a\<^sub>1 b (f k')]])"
\end{lstlisting}

The proof script is not instructive in detail, so we merely note it proceeds by obtaining the set of all indices of chain elements between $a_1$ and $b$.
We can then obtain its maximum $m$ (provided the set is not empty) using Isabelle's \lstinline|Max| operator, and show that $k=m+1$ satisfies the properties we are looking for. We then continue to prove Case (ii), listed below. Notice we already give a suitable definition for the indexing function $g$ in the assumptions.

\begin{lstlisting}
lemma (*for 10*) chain_append_inside:
  assumes long_ch_Y: "[f[a\<^sub>1..a..a\<^sub>n]Y]"
      and Y_def: "X = Y \<union> {b}" "b\<notin>Y"
      and fin_X: "finite X"
      and Yb: "[[a\<^sub>1 b a\<^sub>n]]"
      and k_def: "[[a\<^sub>1 b (f k)]]" "k < card Y"
        "\<not>(\<exists>k'. (0::nat)<k' \<and> k'<k \<and> [[a\<^sub>1 b (f k')]])"
      and g_def: "g = (\<lambda>j::nat.
        if (j\<le>k-1) then f j
        else (if j=k then b else f (j-1)))"
    shows "[g[a\<^sub>1 .. b .. a\<^sub>n]X]"
\end{lstlisting}

We did not manage to split the proof of Case (ii) according to the same conditions seen in Schutz' proof. We argue this is because he restricts his attention to a handful of events only, trusting his reader's intuition to convince them that the ordering of all other events stays the same. We, on the other hand, need to show explicitly that the new way of indexing given by $g$ satisfies the definition of a chain \emph{everywhere} on $X$, i.e.:
\begin{lstlisting}
have "\<forall>n n' n''.
      (finite X \<longrightarrow> n'' < card X) \<and> Suc n = n' \<and> Suc n' = n''
          \<longrightarrow> [[(g n) (g (Suc n)) (g (Suc (Suc n)))]]"
\end{lstlisting}

This means splitting according to the value of the natural number $n$ and its two successors, in order to fix the (conditional) form of the desired indexing function $g$. We do mirror his case splits in the following results, which are all used in different cases according to (the successors of) $n$.

\begin{lstlisting}
have b_middle: "[[(f (k-1)) b (f k)]]" <proof>
have b_right: "[[(f (k-2)) (f (k-1)) b]]" if "k \<ge> 2" <proof>
have b_left: "[[b (f k) (f (k+1))]]" if "k+1 \<le> card Y - 1" <proof>
\end{lstlisting}

\noindent
It may appear that one could force Schutz' case split, but since our definition of \linebreak\hbox{\lstinline|ordering2|} explicitly requires universal quantification over indices, and $g$ is defined piecewise, the case split we employ would still have to be made later on.

The final transformation from a local chain based on the ordering of successive indices to a globally ordered chain $X$ is again precisely the result of Theorem~\ref{thm:2}.
With $g$ now established as a suitable ordering function in these three scenarios, we have completed our proof for the final case of Theorem~\ref{thm:10}. Any finite set of at least two events on a path forms a chain (i.e.\ can be ordered).

\begin{theorem}\label{thm:11}
Any finite set of $N$ distinct events of a path separates it into $N-1$ segments and two prolongations of segments.
\end{theorem}
\begin{proof}
As in the proof of the previous Theorem~\ref{thm:10}, any event distinct from the \linebreak$a_i$ ($i=1,\dots,N$) belongs to a segment (Case (ii)) or a prolongation (Cases (i) and (iii)). Theorem~\ref{thm:1} implies that the $N-1$ segments and two prolongations are disjoint.
\qed
\end{proof}

The final result of Schutz' section 3.6 (Order on a path), Theorem~\ref{thm:11} allows us to use any finite subset of a path in order to split it into disjoint regions. Schutz provides a three-line argument by analogy with the proof of Theorem~\ref{thm:10}, arguing this result is a direct consequence of Theorems~\ref{thm:10} and \ref{thm:1}, employing the same case split as in the proof of the preceding Theorem~\ref{thm:10}. However, we found that
Schutz' statement is unprovable at the point of his stating it. A weaker version can be proven immediately; Schutz' full theorem only becomes true once Theorem~17 can be established. We discuss this issue after defining segments and intervals.

Schutz defines the \emph{segment} between distinct events $a,b$ of a path $ab$ as the set \linebreak$(ab) = \left\lbrace x : \ord{a}{x}{b},\; x \in ab \right\rbrace$. Similarly, he defines the \emph{interval} $|ab|$ as $(ab) \cup \left\lbrace a,b \right\rbrace$, and the \emph{prolongation} of $(ab)$ beyond $b$ as $\left\lbrace x : \ord{a}{b}{x},\; x \in ab \right\rbrace$. In Isabelle, we denote these sets as \lstinline|segment a b|, \lstinline|interval a b|, and \lstinline|prolongation a b| respectively.

Theorem~\ref{thm:11} and its proof sound natural enough to the geometric intuition, taking a path to be somehow line-like. However, the part of the statement regarding the number of segments is impossible to prove at this point.
Given two events $a$ and $b$ on a path $P$, Theorem~\ref{thm:6} (on prolongation, Sec.~\ref{sec:prolongation}) guarantees the existence of $c \in P$ such that $[abc]$ (or alternatively, such that $[cab]$), but we can guarantee the existence of an element $c$ such that $[acb]$ only after Theorem~17 (in Schutz' Chapter~4, not considered here), which states exactly that. Since no such element can be guaranteed to exist, segments can be empty. Then since they are defined as sets, all empty segments are equal (to the empty set), and this degeneracy can reduce the number of segments that exist in the segmentation. The problem is that formally, Theorem~17 relies on Theorem~\ref{thm:13}, which in turn requires Theorem~\ref{thm:11}, so we cannot just postpone this result.

One could fix this problem by taking intervals instead of segments. By definition, no interval is empty, fixing their number as Schutz suggests -- \todo{not double -?} but the intervals would overlap at their endpoints, losing disjointness.
We surmise that one could also prove that there are \textit{at most} $N-1$ segments.
We prove two versions of Theorem~\ref{thm:11}. In one we omit the conclusion about the number of segments (Sec.~\ref{sec:proofs:thm11:nocard}); in the other we include it, but have to assume path density (Sec.~\ref{sec:proofs:thm11:card}).

Ultimately, the problem is not fatal: we do not need to know how many segments there are for the proof of Theorem~\ref{thm:13}, only that a segmentation exists given a chain of events.
The disjointness of the segmentation is also added as a conclusion, while Schutz only mentions it in his proof.

\subsubsection{Without additional assumptions}\label{sec:proofs:thm11:nocard}
One could formalise Schutz' Theorem~\ref{thm:11} faithfully, as a pure existential statement, as in \lstinline|segmentation| given below.
\jtodo{\lstinline{is_pro} has not been mentioned before. Changed it to is\_prolongation. Does that need explaining?}
\begin{lstlisting}
abbreviation disjoint
  where "disjoint A \<equiv> (\<forall>a\<in>A. \<forall>b\<in>A. a \<noteq> b \<longrightarrow> a \<inter> b = {})"
\end{lstlisting}

\enlargethispage{-2\baselineskip}
\begin{lstlisting}
theorem (*11*) segmentation:
  assumes path_P: "P\<in>\<P>"
      and Q_def: "card Q\<ge>2" "Q\<subseteq>P"
    shows "\<exists>S P1 P2. P = ((\<Union>S) \<union> P1 \<union> P2 \<union> Q) \<and>
                     disjoint (S\<union>{P1,P2}) \<and> P1\<noteq>P2 \<and> P1\<notin>S \<and> P2\<notin>S \<and>
                     (\<forall>x\<in>S. is_segment x) \<and>
                     is_prolongation P1 \<and> is_prolongation P2"
\end{lstlisting}

However, in order to show the set of segments $S$ and the two prolongations $P_1$ and $P_2$ exist, and have the desired properties, we have to construct them explicitly. This leads to the more practical theorem \lstinline|show_segmentation|. In fact, this is the statement we prove, and \lstinline|segmentation| can then be derived from it quite easily by using Theorem~\ref{thm:10} to obtain an indexing function $f$ for the set of events $Q$.

\begin{lstlisting}
theorem (*11*) show_segmentation:
  assumes path_P: "P\<in>\<P>"
      and Q_def: "Q\<subseteq>P"
      and f_def: "[f[a..b]Q]"
    fixes P1 defines P1_def: "P1 \<equiv> prolongation b a"
    fixes P2 defines P2_def: "P2 \<equiv> prolongation a b"
    fixes S  defines S_def:
    	 "S \<equiv> if card Q=2 then {segment a b}
              else {segment (f i) (f (i+1)) | i. i<card Q-1}"
    shows "P = ((\<Union>S) \<union> P1 \<union> P2 \<union> Q)" "(\<forall>x\<in>S. is_segment x)"
          "disjoint (S\<union>{P1,P2})" "P1\<noteq>P2" "P1\<notin>S" "P2\<notin>S"
\end{lstlisting}

The additional assumption \lstinline|f_def| turns out to be required in order to follow Schutz' proof of Theorem~\ref{thm:13}, as well as allowing us to give an explicit definition of $S$.
Strictly adhering to Schutz' formulation for Theorem~\ref{thm:11} (like in \hbox{\lstinline|segmentation|}) would lead to additional complexity when proving Theorem~\ref{thm:13} (see Section~\ref{sec:unreach-connected}).


Notice that the definition of $S$ follows our division between short and long chains, and so must the proof.
The case of a short chain $Q=\{a,b\}$ is simple, since $S$ is a singleton with element $(ab)$. All individual required results are deriveable by Isabelle's \textit{sledgehammer} with the exception of $P = (ab) \cup P_1 \cup P_2 \cup Q$, which we prove by translating $x \in P$ into
$\ord{a}{x}{b} \lor \ord{b}{a}{x} \lor \ord{a}{b}{x} \lor x=a \lor x=b$ (by Axiom~\ref{ax:O5}).

For $N\geq3$ we prove $P_1$, $P_2$, and $S$ satisfy the conditions laid out in \hbox{\lstinline|show_segmentation|} one by on via helper lemmas.
The main lemma is that the set $S$ of segments covers the ``inside'' of the chain:

\begin{lstlisting}
lemma int_split_to_segs:
  assumes f_def: "[f[a..b..c]Q]"
    fixes S
  defines S_def: "S \<equiv> {segment (f i) (f(i+1)) | i. i<card Q-1}"
  shows "interval a c = (\<Union>S) \<union> Q"
\end{lstlisting}

The proof is lengthy, but the mechanisation details are largely uninspiring, so we omit these here. It proceeds by finding, for any event $x\in(ac)$, the closest chain events on either side (which give the segment of $S$ containing $x$); conversely, for any event $y\in(\bigcup S) \cup Q$, we apply the betweenness properties of chains as well as overlapping-betweenness lemmas similar to those of Sec.~\ref{sec:order-path} to obtain $\ord{a}{y}{c}$.

Similar lemmas exist for the remaining conclusions of Theorem~\ref{thm:11}, but we omit their proofs too. The main result is the segmentation of the interval: the prolongations just act as a two-sided catch-all for any other element. Furthermore, disjointness of the segments
(of the form \lstinline|segment (f i) (f(i+1))|)
follows from the ordering of finite chains, and obtaining a chain from a finite subset of a path is easy using Theorem~\ref{thm:10}. 

\subsubsection{Assuming path density}\label{sec:proofs:thm11:card}
Since Schutz omitted so many of the conclusions of our own \lstinline|show_segmentation| from his Theorem~\ref{thm:11}, but did insist on the number of segments, we created an additional locale, called \lstinline|MinkowskiDense|, to contain an assumed version of Schutz' Theorem~17. This is safer than a sorried theorem (see Sec.~\ref{sec:bgr:isa}) -- the assumption \lstinline|path_dense| will never be used accidentally, as long as we never work in the locale \lstinline|MinkowskiDense|, or in a locale built on top of it. We prove that the cardinality of the set $S$ of segments in the theorem \lstinline|show_segmentation| is indeed $N-1$ if path density is assumed.

\jeptodo{Do we address MinowskiDense somewhere else?}
\begin{lstlisting}
locale MinkowskiDense = MinkowskiSpacetime +
  assumes path_dense: "path ab a b \<Longrightarrow> \<exists>x. [[a x b]]"
begin

lemma segment_nonempty:
  assumes "path ab a b"
  obtains x where "x \<in> segment a b"
  using path_dense by (metis abc_abc_neq seg_betw assms)
\end{lstlisting}

The number-of-segments statement is obviously only interesting if $N\geq3$, which simplifies the definition of $S$. The remaining conditions are those of the helper lemmas for Theorem~\ref{thm:11}. Schutz' ``$N-1$ segments'' turns into a proposition on the cardinality of the set of segments $S$.

\begin{lstlisting}
lemma number_of_segments:
  assumes path_P: "P\<in>\<P>"
      and Q_def: "Q\<subseteq>P"
      and f_def: "[f[a..b..c]Q]"
    shows "card {segment (f i) (f(i+1)) | i. i<(card Q - 1)}
    	   = card Q - 1"
\end{lstlisting}

We can show two sets have equal cardinality if a bijection exists between them.\footnote{This is generally taken as a definition in mathematics (e.g. Liebeck \cite[p.~185]{liebeck2011}. Isabelle's definition is more technical, but the proof strategy still applies.}
To this end we define a function $g\colon i \mapsto (Q_iQ_{i+1})$, and prove it is a bijection between the sets $I=\{0\;..\;N-2\}$ and $S$.
With Isabelle's functions being total over types (in the case of $g$, total over $\setN$, not $I$), we must be subtle about what we prove: not bijectivity of $g$, but only bijectivity of its restriction to $I$.
This is expressed using \lstinline|bij_betw| in Isabelle.
In the listing below, $N = |Q|$, and the direct image of a function applied to a set is denoted by a backtick.

\begin{lstlisting}
proof -
  let ?g = "\<lambda> i. segment (f i) (f (i+1))"
  have "?g ` {0..?N-2} = ?S" <proof>
  moreover have "inj_on ?g {0..?N-2}" <proof>
  ultimately have "bij_betw ?g {0..?N-2} ?S"
    using inj_on_imp_bij_betw by auto
  thus ?thesis
    using assms(5) bij_betw_same_card
    by (metis (no_types, lifting) (*...*))
qed
\end{lstlisting}

Diving briefly into the proof of injectivity, we show where path density comes into play. Injectivity is proven as usual, that for
$i,j \in I$,
we have $g(i) = g(j) \Longrightarrow i=j$. This is shown by contradiction (\lstinline|assume "i\<noteq>j"|),
then split into the cases seen in Fig.~\ref{fig:thm11_j_cases}.
Notice this is almost the case split of Theorem~\ref{thm:10}, which is perhaps the reference Schutz makes to the preceding proof. Picking the left-most case of Fig.~\ref{fig:thm11_j_cases} as an example, such that $\ord{f(i+1)}{f(i)}{f(j)}$, we use \lstinline|segment_nonempty|
to obtain an element $e$ that satisfies the contradictory orderings $\ord{e}{f(j)}{f(j+1)}$ and $\ord{f(j)}{e}{f(j+1)}$.

\jeptodo{Fig. 5 caption: are we lower-casing `theorem' and `section' when referring to Schutz and upper-casing them when referring to this paper? I've noticed a few lower-cases for both.}
\begin{figure}
\centering
\includegraphics[width=\textwidth]{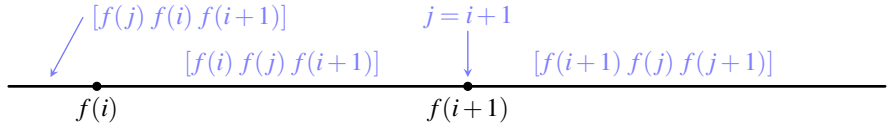}
\caption{\label{fig:thm11_j_cases}Four case splits in the proof of \lstinline{number_of_segments} for \lstinline{MinkowskiDense}'s Theorem~\ref{thm:11}, according to the index $j$, respectively the event $f(j)$.}
\end{figure}

\begin{lstlisting}
          assume "[[(f(i+1)) (f i) (f j)]]"
          then obtain e where "e\<in>?g i" using segment_nonempty
            by (metis (*...*))
          hence "[[e (f j) (f(j+1))]]"
            using \<open>[[(f(i+1)) (f i) (f j)]]\<close> <proof>
          moreover have "e\<in>?g j"
            using \<open>e \<in> ?g i\<close> asm(3) by blast
          ultimately show False
            by (simp add: abc_only_cba1 seg_betw)
\end{lstlisting}

\noindent
The other cases run similarly.
The case of $\ord{f(i)}{f(j)}{f(i+1)}$ proceeds in one step, since the assumption $g(i) = g(j)$ then gives $\ord{f(j)}{f(j)}{f(j+1)}$, which is false by Axiom~\ref{ax:O3}.

%% file: in/thms-3.7.tex
As mentioned in the introduction to Section~\ref{sec:ch3}, Theorem 12 is not included in this formalisation. See Section~\ref{sec:conclusion} for a brief discussion.

%% file: in/thms-3.8.tex
Since it was impossible to prove the full version of Schutz' Theorem~\ref{thm:11}, one may wonder if Schutz' results relying on this theorem remain valid.
As laid out in Sec.~\ref{sec:order-path}, the part of Theorem~\ref{thm:11} formalised in \lstinline|number_of_segments| relies indirectly on Theorem~\ref{thm:13}. Thus, mechanising Theorem~\ref{thm:13} while using only the weaker, verified version of Theorem~\ref{thm:11}, serves to dispel any doubts.
We give the prose statement and proof below, and follow them with the formalised theorem. 
Schutz here introduces a new notation for ``non-strict ordering'' \cite[p.~27]{schutz1997}:
$$[a \; b \; c ]\!] \equiv \ord{a}{b}{c} \text{ or } b = c\;.$$

\begin{theorem}[Connectedness of the Unreachable Set]\label{thm:13}
\\Given any path $Q$, any event $b\notin Q$, and distinct events $Q_x, Q_z \in Q(b,\emptyset)$, then
$$\ord{Q_x}{Q_y}{Q_z} \Longrightarrow Q_y \in Q(b, \emptyset)\;.$$
\end{theorem}
\begin{proof}
By axiom I6 there is a finite chain $[Q_0 \; Q_1 \; \dots \; Q_{n-1} \; Q_n]$
(where $Q_0 = Q_x$ and $Q_n = Q_z$) so Theorem 11 implies that for some $i \in \left\lbrace 1, \dots, n \right\rbrace$, $[Q_{i-1} \; Q_y \; Q_i ]\!]$ whence axiom I6 implies that $Q_y \in Q(b, \emptyset)$.
\qed
\end{proof}



\begin{lstlisting}
theorem (*13*) unreach_connected:
  assumes path_Q: "Q\<in>\<P>"
      and event_b: "b\<notin>Q" "b\<in>\<E>"
      and unreach: "Q\<^sub>x \<in> \<emptyset> Q b" "Q\<^sub>z \<in> \<emptyset> Q b" "Q\<^sub>x \<noteq> Q\<^sub>z"
      and xyz: "[[Q\<^sub>x Q\<^sub>y Q\<^sub>z]]"
    shows "Q\<^sub>y \<in> \<emptyset> Q b"
\end{lstlisting}

We follow Schutz' proof at the start, obtaining a chain on $Q$ from Axiom~\ref{ax:I6}. We call this chain $X$, with indexing function $f$, while Schutz distinguishes the chain $\{Q_i\}_{i=0\dots n}$ from the path $Q$ only by the subscripts.

\begin{lstlisting}
  obtain X f where X_def: "ch_by_ord f X"
      "f 0 = Q\<^sub>x" "f (card X - 1) = Q\<^sub>z"
      "\<forall>i\<in>{1 .. card X - 1}.
        f i \<in> \<emptyset> Q b \<and>
        (\<forall>Qy\<in>\<E>. [[f(i-1) Qy f i]] \<longrightarrow> Qy \<in> \<emptyset> Q b)"
      "short_ch X \<longrightarrow> Q\<^sub>x \<in> X \<and> Q\<^sub>z \<in> X \<and>
         	      (\<forall>Q\<^sub>y\<in>\<E>. [[Q\<^sub>x Q\<^sub>y Q\<^sub>z]] \<longrightarrow> Q\<^sub>y \<in> \<emptyset> Q b)"
    using I6 [OF assms(1-6)] by blast
  hence "[f[Q\<^sub>x..Q\<^sub>z]X]" <proof>
\end{lstlisting}

\noindent
The final line above makes certain $Q_x$ and $Q_z$ (defined via their indices only) are indeed the bounding events of the chain.
It is only at this stage that we realised Axiom~\ref{ax:I6} had to be modified to account for the case of short chains. This is an artifact of our chain definition (see Sec.~\ref{sec:axioms}), where the indexing function $f$ only has meaning for chains of at least three elements (strictly speaking, the same is true of Schutz' prose definition). Thus, for example, if $|X|=2$, we cannot conclude $f(0)\in X$.




We split the remainder of the proof according to whether the obtained chain $X$ is long or short (i.e.\ whether there is a meaningful indexing function $f$). The case of short chains is very straightforward, and the theorem follows immediately from the short-chain clause of Axiom~\ref{ax:I6} (invoked as \lstinline|X_def(5)|, the fifth fact listed under the name \lstinline|X_def| above).

\begin{lstlisting}
  show ?thesis
  proof cases
    assume "N=2"
    thus ?thesis
      using X_def(1,5) xyz \<open>N = card X\<close> event_y short_ch_card_2
      by auto
  next
    assume "N\<noteq>2"
    hence "N\<ge>3" using \<open>2 \<le> N\<close> by auto
    have y_cases: "Q\<^sub>y\<in>X \<or> Q\<^sub>y\<notin>X" by blast
\end{lstlisting}

\noindent
A second layer of case splitting occurs only in the case of $N \geq 3$, and is given in the fact \lstinline|y_cases| (final line above). Schutz absorbs this split into the \textit{non-strict ordering} defined at the beginning of this section.
He then relies on his reader to consider both cases and to dispense with the (often degenerate\jeptodo{Often degenerate in this particular proof or generally?}) $Q_y \in X$ case.
Isabelle would not accept such an implicit approach, so this non-strict notation is not used for formalising Theorem~\ref{thm:13}, and we treat both cases explicitly.

Now that we have dealt with short chains, it is time to do as Schutz suggests, and use Theorem~\ref{thm:11}. In fact, we do not need the entirety of the theorem, but only the part of the result relating to the segmentation of an interval on a path, i.e.\ \lstinline|int_split_to_segs| (see Sec.~\ref{sec:order-path}).

To keep the proof as simple as possible, it is vital that the set of events $Q$ is already indexed as a chain. To see why, assume we have no indexing function, but only a set of events.
Then \lstinline|segmentation| (see Sec.~\ref{sec:order-path}) does provide a set of segments, but we have no handle on their endpoints: in particular, there is no proof that the segments are made up of events that are adjacent according to the ordering $f$. An early version of the proof of Theorem~\ref{thm:13} did go this route, using the interesting uniqueness result \lstinline|chain_unique_upto_rev| to relate a chain obtained from the segment endpoints to the chain $X$. With the more explicit formulation of Theorem~\ref{thm:11}, \lstinline|show_segmentation|, this extra complexity disappears.


\begin{lstlisting}
lemma (in MinkowskiSpacetime) chain_unique_upto_rev:
  assumes "[f[a..c]X]" "[g[x..z]X]" "card X \<ge> 3" "i < card X"
  shows "f i = g i \<or> f i = g (card X - i - 1)"
\end{lstlisting}


If $Q_y$ is an event of the chain $X$, I6 immediately implies $Q_y \in Q(b,\emptyset)$ (this is fact \hbox{\lstinline|X_def(4)|).}
If not, i.e.\ $Q_y \notin X$, we obtain the relevant segment from Theorem~\ref{thm:11} much like Schutz does in prose.

More specifically, we find the index $i$
such that $\ord{f(i-1)}{Q_y}{f(i)}$,
and prove our goal $Q_y \in Q(b,\emptyset)$. What follows is just a listing of the most salient statements of the remaining proof. The set $S$ is defined exactly as in \lstinline|show_segmentation|. Once the index $i$ is shown to exist, the result follows from a simple application of, again, Axiom~\ref{ax:I6} in the guise of \lstinline|X_def(4)|.

\jeptodo{I've made this not split over two pages. Check to make sure it hasn't broken anything, and apologies in advance :) Fig. 6 is further along now I think.}
\vbox{
\begin{lstlisting}
      assume "Q\<^sub>y \<notin> X"
      have "Q\<^sub>y \<in> \<Union>?S"
        using int_split_to_segs [OF `[f[Q\<^sub>x..c..Q\<^sub>z]X]`] <proof>
      (*...*)
      obtain i where i_def: "i\<in>{1..(card X)-1}"
                            "[[(f(i-1)) Q\<^sub>y (f i)]]"
        by blast
      show ?thesis
        by (meson X_def(4) i_def)
\end{lstlisting}
}

The completion of this proof demonstrates several benefits of mechanisation of formal mathematics. First, resolution of a minor lapse in the prose led to a proof of a result not found in the original text, \lstinline|chain_unique_upto_rev|. This is interesting in its own right, as it generalises Theorem~\ref{thm:1} to chains much like \lstinline|chain_sym| generalises Axiom~\ref{ax:O2}.
Secondly, we were able to reconcile a follow-on result with a necessarily weaker version of the required Theorem~\ref{thm:11} (\lstinline|show_segmentation|).

\begin{theorem}[Second Existence Theorem]\label{thm:14}
\begin{enumerate}
    \item[(i)] Given a path $Q$ and a pair of events $a,b \notin Q$, each of which can be joined to $Q$ by some path, there are events $y,z \in Q$ such that
    $$\ord{y}{Q(a,\emptyset)}{z} \text{ and } \ord{y}{Q(b,\emptyset)}{z}\;.$$
    \item[(ii)] Given a path $Q$ and a pair of events $a,b \notin Q$ each of which can be joined to $Q$ by some path and a pair of events $c,d \in Q$, there is an event $e \in Q$ and paths $ae$, $be$ such that $\ord{c}{d}{e}$.
    \item[(iii)] Given two paths $Q$ and $R$ which meet at $x$, an event $a \in R \setminus \lbrace x \rbrace$ and an event $b \notin Q$ which can be joined to $Q$ by some path, there is an event $e$ and paths $ae$, $be$ such that $\ord{x}{Q(a,\emptyset)}{e}$.
\end{enumerate}
\end{theorem}

\noindent
The betweenness relation is here extended to sets of events: for a set $S$,
\[\ord{a}{S}{b} \iff \forall x \in S: \ord{a}{x}{b}\;.\]
The First Existence Theorem (Theorem~\ref{thm:5}) provides the basic geometric setup for the proofs of Theorem~\ref{thm:6} and the important Lemma~\ref{lem:1} (leading to Theorems~\ref{thm:9} and \ref{thm:10}). Using several results of Chapter~3, which it concludes, Theorem~\ref{thm:14} provides similar constructions for use in the geometric proofs of subsequent chapters. A visualisation of parts (i) and (iii) is provided in Fig.~\ref{fig:thm14} (part (ii) is similar to (i)).

\begin{figure}
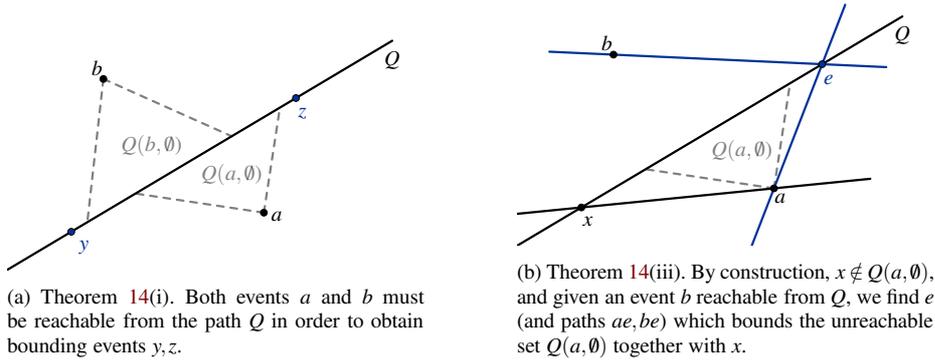

    \centering
     \begin{subfigure}[b]{0.45\textwidth}
         \centering
         \includegraphics[width=\textwidth]{tikz/thm14i.tex}
         \caption{Theorem \ref{thm:14}(i). Both events $a$ and $b$ must be reachable from the path $Q$ in order to obtain bounding events $y,z$.}
         \label{fig:thm14i}
     \end{subfigure}
     \hfill
     \begin{subfigure}[b]{0.45\textwidth}
         \centering
         \includegraphics[width=\textwidth]{tikz/thm14iii.tex}
         \caption{Theorem \ref{thm:14}(iii). By construction, $x\notin Q(a,\emptyset)$, and given an event $b$ reachable from $Q$, we find $e$ (and paths $ae, be$) which bounds the unreachable set $Q(a,\emptyset)$ together with $x$.}
         \label{fig:thm14iii}
     \end{subfigure}
    \caption{Visualisation of Theorem~\ref{thm:14}.}
    \label{fig:thm14}
\end{figure}

Schutz' proofs for each of the three statements are short \cite[p.~30]{schutz1997}, and we will reproduce them here to highlight the differences in our formalisation. By far the most different is the proof for part (i).

\begin{lstlisting}
theorem (*14i*) second_existence_thm_1:
  assumes path_Q: "Q\<in>\<P>"
      and events: "a\<notin>Q" "b\<notin>Q"
      and reachable: "path_ex a q1" "path_ex b q2" "q1\<in>Q" "q2\<in>Q"
    shows "\<exists>y\<in>Q.\<exists>z\<in>Q. (\<forall>x\<in>\<emptyset> Q a. [[y x z]]) \<and>
                      (\<forall>x\<in>\<emptyset> Q b. [[y x z]])"
\end{lstlisting}

\begin{proof}[Theorem~\ref{thm:14}(i)]
Theorem~\ref{thm:4} implies that both sets $Q(a,\emptyset)$ and $Q(b,\emptyset)$ are bounded in both directions by events which do not belong to the unreachable sets themselves, so the union $Q(a,\emptyset) \cup Q(b,\emptyset)$ is bounded by distinct events $y$, $z$ which do not belong to the union of the unreachable sets.
\qed
\end{proof}

In the proof above, Schutz implicitly extends his notion of boundedness to sets. We assume that he means a similar property as he did for chains, i.e. using strict betweenness. We take a set of events $S$ to be bounded by $a$, $b$ if $\ord{a}{S}{b}$, or equivalently $\forall x \in S: \ord{a}{x}{b}$, and we will keep this explicit in our formalisation.%
\footnote{A case can be made that Schutz means a non-strict bound when talking about bounded sets. This would be corroborated by his insistence that the bounds do not belong to the unreachable sets, which would be an immediate consequence of Axiom~\ref{ax:O3} for strict bounds, rather than a conclusion of Theorem~\ref{thm:4}. It would also explain why boundedness of sets is used in the proof, while the explicit betweenness is given in the theorem's conclusion. Nonetheless, for consistency, we stick with the meaning that aligns with the bounds of a chain.}

At the top level, we are able to follow Schutz' proof quite well. We give a truncated listing below. Two difficulties arise, one for each fact in the listing.

\begin{lstlisting}
proof -
  have get_bds: "\<exists>la\<in>Q. \<exists>ua\<in>Q. la \<notin> \<emptyset> Q a \<and> ua \<notin> \<emptyset> Q a \<and>
                               (\<forall>x \<in> \<emptyset> Q a. [[la x ua]])"
    if asm: "a\<notin>Q" "path_ex a q" "q\<in>Q"
   for a
  <proof>

  have "\<exists>y\<in>Q. \<exists>z\<in>Q. (\<forall>x\<in>(\<emptyset> Q a)\<union>(\<emptyset> Q b). [[y x z]])"
  <proof>

  show ?thesis <proof>
qed
\end{lstlisting}

First, to prove the fact \lstinline|get_bds|, we require Theorem~\ref{thm:13}, not just Theorem~\ref{thm:4} as proposed by Schutz. This is because Theorem~\ref{thm:4} gives us, for fixed $Q_x \in Q \setminus Q(b,\emptyset)$, one event $Q_z$ for every $Q_y \in Q(b,\emptyset)$ such that $\ord{Q_x}{Q_y}{Q_z}$. Showing that there is a single $Q_z$ that bounds all possible $Q_y$ requires a proof by contradiction that falsifies Theorem~\ref{thm:13} (Connectedness of the Unreachable Set). Schutz' proof misses this aspect.

Second, to prove a union of bounded sets is bounded, we introduce a way to abstract case splits over the many possible orderings of four events. Thinking about this proof as a mathematician, it is clear what happens: there are two bounds for each set, one on each side, and no matter what the relationship of the sets to one another is, there are always two bounds that qualify as bounds of the union. However, this reasoning breaks down into many case splits in Isabelle, because once we obtain bounds explicitly, we need to consider any possible ordering of all four of them, as well as the possibility of any subset of them being equal.

This leads to a lemma that encapsulates the idea of assuming some ordering ``without loss of generality'' \cite{harrison2009,scott2015,scott2016}, by using the symmetries of the betweenness relation. This can be compared to standard results, e.g. \lstinline|linorder-wlog| and \lstinline|linorder-less-wlog| in Isabelle/HOL's \lstinline|Orderings| theory.
We only list the lemma in the case of distinct events here. A similar result is provided to deal with the possible cases of equality between four events. Since we deal with orderings of four events, the notation for betweenness is extended to apply in the quaternary case (this is equivalent to saying a chain exists such that the four elements can be indexed left-to-right).

\begin{lstlisting}
lemma (in MinkowskiSpacetime) wlog_endpoints_distinct:
  assumes path_A: "A\<in>\<P>"
      and symmetric_Q: "\<And>a b I. Q I a b \<Longrightarrow> Q I b a"
      and Q_implies_path: "\<And>a b I. \<lbrakk>I\<subseteq>A; Q I a b\<rbrakk> \<Longrightarrow> b\<in>A \<and> a\<in>A"
      and symmetric_P: "\<And>I J. \<lbrakk>\<exists>a b. Q I a b; \<exists>a b. Q J a b; P I J\<rbrakk>
                             \<Longrightarrow> P J I"
      and "\<And>I J a b c d. \<lbrakk>Q I a b; Q J c d; I\<subseteq>A; J\<subseteq>A;
                [[a b c d]] \<or> [[a c b d]] \<or> [[a c d b]]\<rbrakk> \<Longrightarrow> P I J"
  shows "\<And>I J a b c d. \<lbrakk>Q I a b; Q J c d; I\<subseteq>A; J\<subseteq>A;
                a\<noteq>b \<and> a\<noteq>c \<and> a\<noteq>d \<and> b\<noteq>c \<and> b\<noteq>d \<and> c\<noteq>d\<rbrakk> \<Longrightarrow> P I J"
\end{lstlisting}

This lemma takes two generic predicates: $P$, a function of two sets of events (e.g. ``the union of these sets is bounded''), and $Q$, a function of two events and a set (e.g. ``this set is the interval between these two events''). The premise $R$ quantifies over two events in $Q$ (e.g. ``this set is an interval'').
For any such relations, the final assumption of the lemma \lstinline|wlog_endpoints_distinct| gives the only essentially distinct cases: only three orderings of four events need to be considered given that the relations $P,Q$ are symmetric, and that we know betweenness is symmetric by Axiom~\ref{ax:O2}.%
\footnote{The condition \lstinline|Q_implies_path| is needed to establish that some ordering exists, via Axiom~\ref{ax:O5}.}
We are then allowed to conclude that all possible orderings follow from the three essentially distinct ones. Notice how the naming of events is left unspecified using Isabelle's universal meta-quantification $\bigwedge$.\footnote{This may be compared to Schutz' formulation of Theorem~\ref{thm:10}, ``any set of $n$ distinct events can be represented by the notation $a_1,a_2,\dots,a_n$ such that [\dots]''.}

Thus we can, for example, prove boundedness of the union of two bounded sets without considering all possible orderings (24, if we don't use Axiom~\ref{ax:O2}). In exchange, the lemma is very verbose, but still remains specific to using betweenness and paths: generalising those would lead to an even more unwieldy statement. Using \lstinline|wlog_endpoints_distinct| is best done by splitting the proof of $P I J$ (for concrete $P$) immediately using Isabelle's \lstinline|rule_tac|, and then prove the lemma's assumptions one by one, fixing variables inside the scope of each subgoal.

\begin{lstlisting}
theorem (*14*) second_existence_thm_2:
  assumes path_Q: "Q\<in>\<P>"
      and events: "a\<notin>Q" "b\<notin>Q" "c\<in>Q" "d\<in>Q" "c\<noteq>d"
      and reachable: "\<exists>P\<in>\<P>. \<exists>q\<in>Q. path P a q"
                     "\<exists>P\<in>\<P>. \<exists>q\<in>Q. path P b q"
    shows "\<exists>e\<in>Q. \<exists>ae\<in>\<P>. \<exists>be\<in>\<P>. path ae a e
                             \<and> path be b e \<and> [[c d e]]"
\end{lstlisting}

After the hard work of part (i), the second statement of Theorem~\ref{thm:14} is easier to prove, as we don't need WLOG results, and can rely on the first part of the theorem to provide the general setup. In fact, we follow Schutz' proof with little trouble.

\begin{proof}[Theorem~\ref{thm:14}(ii)]
In the cases where $\ord{c}{d}{y}$ or $\ord{c}{d}{z}$ we define $e$ to be $y$ or $z$ respectively. The other cases are where
$([\![y\;c\;d] \text{ or } [c\;y\;d]\!]) \text{ and } ([\![z\;c\;d] \text{ or } [c\;z\;d]\!])$:
in these cases the Prolongation Theorem (Th.~\ref{thm:6}) implies the existence of an event $e$ such that $\ord{c}{d}{e}$ and by theorem~\ref{thm:10} the event $e$ is not between the bounding events $y$, $z$ so there are paths $ae$, $be$.
\qed
\end{proof}

The main case split according to orderings of the events $c,d,y,z$ can be found in our formalisation as well. We omit several facts establishing the situation resulting from applying part (i), which Schutz implicitly continues from. The two cases $\ord{c}{d}{y}$ and $\ord{c}{d}{z}$ are solved in a few steps, and are of no great interest, so we give only the final case.


\begin{lstlisting}
proof -
  (*...*)
  let ?P = "\<lambda>e ae be. (e\<in>Q \<and> path ae a e \<and> path be b e \<and> [[c d e]])"
  
  have "[[c d y]] \<or> [[c d z]] \<or>
        ((\<lbrakk>y c d]] \<or> [[c y d\<rbrakk>) \<and> (\<lbrakk>z c d]] \<or> [[c z d\<rbrakk>))"
  <proof>
  thus ?thesis
  proof (rule disjE3)
    (*...*)
    assume "(\<lbrakk>y c d]] \<or> [[c y d\<rbrakk>) \<and> (\<lbrakk>z c d]] \<or> [[c z d\<rbrakk>)"
    
    have "\<exists>e. [[c d e]]" <proof>
    then obtain e where "[[c d e]]" by auto
    
    have "\<not>[[y e z]]" <proof>
    (*...*)
    thus ?thesis
      using \<open>[[c d e]]\<close> \<open>e \<in> Q\<close> by blast
  qed
qed
\end{lstlisting}

Isabelle's \emph{sledgehammer} can automatically construct a proof for $\exists e.\; \ord{c}{d}{e}$ that uses only Theorem~\ref{thm:6}.
Similarly, a proof for $\lnot \ord{y}{e}{z}$ can be found.
We do not need Theorem~\ref{thm:10}, as in Schutz' proof. Instead, we use smaller lemmas specific to orderings of only three events, which Isabelle handles with greater ease, particularly in the presence of non-strict ordering.

\begin{lstlisting}
theorem (*14*) second_existence_thm_3:
  assumes paths: "Q\<in>\<P>" "R\<in>\<P>" "Q\<noteq>R"
      and events: "x\<in>Q" "x\<in>R" "a\<in>R" "a\<noteq>x" "b\<notin>Q"
      and reachable: "\<exists>P\<in>\<P>. \<exists>q\<in>Q. path P b q"
    shows "\<exists>e\<in>\<E>. \<exists>ae\<in>\<P>. \<exists>be\<in>\<P>. path ae a e \<and> path be b e \<and>
                               (\<forall>y\<in>\<emptyset> Q a. [[x y e]])"
\end{lstlisting}

\begin{proof}[Theorem~\ref{thm:14}(iii)]
By (ii) above, if we let $c := x$ and take any $d \in Q(a,\emptyset)$ there is an event $e \in Q$ and paths $ae$, $be$ such that $\ord{x}{d}{e}$. Theorem~\ref{thm:13} then implies that $\ord{x}{Q(a,\emptyset)}{e}$.
\qed
\end{proof}

Again, the formalisation of part (iii) follows Schutz' proof closely. The events $d,e$ in his first sentence can be obtained automatically again. We do need to consider a proof by contradiction and several case splits to prove $\ord{x}{Q(a,\emptyset)}{e}$, namely for $y \in Q(a,\emptyset)$, the non-trivial cases to be falsified are $\ord{y}{x}{e} \lor \ord{x}{e}{y}$. In both cases we use Theorem~\ref{thm:13} as the only major result.

